\documentclass[12pt, journal, draftclsnofoot, onecolumn]{IEEEtran}
\usepackage{lettrine}
\usepackage[cmex10]{amsmath}
\usepackage{amssymb}
\usepackage{amsthm}
\usepackage{cite}
\usepackage{amsmath}
\usepackage{amssymb}
\usepackage{booktabs}
\usepackage{balance}
\usepackage{subfigure}
\usepackage{setspace}
\usepackage{multicol}
\usepackage{graphicx}
\usepackage{epsfig}
\usepackage{epstopdf}
\usepackage{cases}
\usepackage{algorithmic}
\usepackage[ruled,vlined]{algorithm2e}
\usepackage{bm}
\usepackage{amsmath, amssymb, amsthm}
\usepackage{upgreek}
\usepackage{CJK}
\usepackage{multirow}
\usepackage{makecell}
\usepackage{float}
\usepackage{color}
\usepackage{bm}
\usepackage{tabularx}

\theoremstyle{definition} 
\theoremstyle{definition} 
\theoremstyle{plain} \newtheorem{lemma}{Lemma}
\theoremstyle{plain} \newtheorem{proposition}{Proposition}

\usepackage{mathrsfs}
\usepackage{threeparttable}
\usepackage{ifpdf}

\ifCLASSINFOpdf

\else

\fi

\begin{document}
%

\title{Reliability Enhancement for VR Delivery in Mobile-Edge Empowered  Dual-Connectivity \textcolor{blue}{Sub-6 GHz} and mmWave  HetNets }

\author{Zhuojia~Gu, 
        Hancheng~Lu,
        Peilin~Hong,
        and Yongdong~Zhang

\IEEEcompsocitemizethanks{
\IEEEcompsocthanksitem This work was supported in part by the National Science Foundation of China (No.61771445, No.91538203).
\IEEEcompsocthanksitem Zhuojia Gu, Hancheng Lu, Peilin Hong are with CAS Key Laboratory of Wireless-Optical Communications, University of Science and Technology of China, Hefei 230027, China (email: guzj@mail.ustc.edu.cn; hclu@ustc.edu.cn; plhong@ustc.edu.cn).
\IEEEcompsocthanksitem Yongdong Zhang is with the Department of Electronic Engineering and Information Science (EEIS), University of Science and Technology of China, Hefei 230027, China (email: zhyd73@ustc.edu.cn).


}
}

\maketitle

%


\begin{abstract}
The reliability of current virtual reality (VR) delivery is low due to the limited resources on VR head-mounted displays (HMDs) and the transmission rate bottleneck of sub-6 GHz networks. In this paper, we propose a dual-connectivity \textcolor{blue}{sub-6 GHz} and mmWave heterogeneous network architecture empowered by mobile edge capability. The core idea of the proposed architecture is to utilize the complementary advantages of \textcolor{blue}{sub-6 GHz} links and mmWave links to conduct a collaborative edge resource design, which aims to improve the reliability of  VR delivery. From the perspective of stochastic geometry, we analyze the reliability of VR delivery and theoretically demonstrate that \textcolor{blue}{sub-6 GHz} links can be used to enhance the reliability of VR delivery despite the large mmWave bandwidth. Based on our analytical work, we formulate a joint caching and computing optimization problem with the goal to maximize the reliability of VR delivery. By analyzing the coupling caching and computing strategies at HMDs, \textcolor{blue}{sub-6 GHz} and mmWave base stations (BSs), we further transform the problem into a multiple-choice multi-dimension knapsack problem. A best-first branch and bound algorithm and a difference of convex programming algorithm are proposed to obtain the optimal and sub-optimal solution, respectively. Numerical simulations demonstrate the performance improvement using the proposed algorithms, and reveal that caching more monocular videos at \textcolor{blue}{sub-6 GHz} BSs and more stereoscopic videos at mmWave BSs can improve the VR delivery reliability efficiently.
\end{abstract}

\vspace{-3mm}


%
\IEEEpeerreviewmaketitle

\section{Introduction}
In recent years, the interest of virtual reality (VR) in academia and industry has been unprecedented \cite{hu2020cellular}.
To achieve the immersive experience, the most important task is to increase the resolution of VR applications to the resolution of  human eyes \cite{elbamby2018toward}.
It is impractical to cache all rendered VR videos locally at VR head-mounted displays (HMDs) in advance, especially in the scenarios where user interactions are required. Moreover, users prefer to experience VR videos anytime and anywhere, compared with being tied down by wired cables. Therefore, VR videos are expected to be delivered  real-time over wireless networks.
However, the bandwidth required for delivering  VR videos is usually 4-5 times the bandwidth required for delivering  conventional high-definition videos \cite{taghavinasrabadi2017adaptive},
which puts tremendous pressure on the network bandwidth. On the other hand, VR applications are typical ultra-reliable low-latency communication (URLLC) applications, and the end-to-end delay exceeding 20 ms can cause dizziness of users \cite{han2019mobile}. Thus, it is essential to ensure the high reliability of VR delivery networks, which means that ensuring more data packets delivered to  HMDs  within the latency requirement of VR applications for a satisfactory VR viewing experience. In this regard, two fundamental problems for  VR delivery should be addressed: \textit{1) How to enhance the reliability of VR delivery when delivering a large amount of video data over wireless networks,} and \textit{  2) how to enhance the reliability of VR delivery  when  performing time-consuming  projecting and rendering of raw VR viewpoints?}

For the first problem, using the current sub-6 GHz  network is limited by the bottleneck of wireless bandwidth and cannot meet the URLLC requirements of VR applications.
The sufficient spectrum resources brought by millimeter wave (mmWave) communication are considered to be the key enablers of 5G applications. Sufficient bandwidth resources make it possible to transmit large-capacity VR videos in real time. 
However, mmWave signals are prone to be blocked and suffer severe fading. This poses a great challenge to  VR applications, because users may frequently experience blockage caused by buildings, human bodies, and environmental facilities. To tackle the problem, dual-connectivity (DC) \textcolor{blue}{sub-6 GHz} and mmWave heterogeneous networks (HetNets) are promising for enhancing the reliability of VR delivery. DC is an appealing technique for performance enhancement in wireless networks, and it allows the user to have two simultaneous connections and utilize both radio resources \cite{kibria2018stochastic}. Inherently, DC \textcolor{blue}{sub-6 GHz} and mmWave HetNets utilize the complementary advantages of the wide signal coverage in \textcolor{blue}{sub-6 GHz} networks and sufficient spectrum resources  in mmWave networks, which is suitable for reliability enhancement of VR delivery.

For the second problem,
computation offloading is seen as a key enabler to provide the required video projection rendering capabilities for VR delivery. Edge computing servers are suitable for performing high CPU- and GPU-intensive computing tasks. By providing efficient computing resources close to users, mobile edge computing (MEC) strikes a balance between communication delay and computing delay. Specifically, the HMD can upload the tracking information (e.g., game actions or viewpoint preferences) to the MEC server to offload computing tasks. In this case, the MEC server uses the computing resources to project monocular videos (MVs) into \textcolor{blue}{stereoscopic videos (SVs)}, and sends the downlink SVs to the HMD. In addition, the caching capability of MEC servers can help save the projection and rendering time.

Based on the above discussion, in this paper, we propose a mobile-edge empowered DC \textcolor{blue}{sub-6 GHz} and mmWave HetNet architecture, and aim to conduct a collaborative design of edge network resources to enhance the reliability of VR delivery.

\subsection{Related Work}
\textit{1) Reliability Enhancement for Wireless Transmission:}
Ultra-reliable communication has become the vital support of mission-critical 5G applications. Mei \textit{et al.} \cite{mei2018latency} proposed  a reliability guaranteed  resource allocation scheme in vehicle-to-vehicle (V2V) networks. Guo \textit{et al.} \cite{guo2019resource} extended the research to perform a reliability-aware resource allocation in V2V networks based only on large-scale fading channel information. Popovski \textit{et al.} \cite{popovski2019wireless} analyzed the fundamental tradeoffs between ultra-reliability and some other metrics (i.e., latency, bandwidth occupancy and energy consumption) for wireless access.
Recently, the joint application of  \textcolor{blue}{sub-6 GHz} and mmWave in HetNets has been proposed as  an attractive solution to improve the network performance \cite{elshaer2016downlink}.
Semiari \textit{et al.} \cite{semiari2017joint, shi2018downlink} introduced DC mode into \textcolor{blue}{sub-6 GHz} and mmWave HetNets to ensure data transmission reliability while considering the user association and scheduling. DC implements carrier aggregation between sites through base station coordination to realize diversified association, which can reduce the handover failure and wireless link failure, thus improving the transmission reliability.
DC mode has been extended to multi-connectivity (MC) mode  in \cite{she2018improving, wolf2018reliable}, which enables simultaneous connections of multiple  air interfaces such as cellular, device-to-device (D2D) and WiFi links to further support ultra-reliable applications.


\textcolor{blue}{
\textit{2) DC-assisted mobile edge computing:} Some prior works proposed to use DC technology for mobile edge computing to achieve energy-efficient computation offloading performance. Guo \textit{et al.} \cite{guo2018computation} studied the problem of computation offloading for multi-access mobile edge computing in ultra-dense networks to reduce the overall energy consumption. Guo \textit{et al.} \cite{guo2018mobile} also extended the scenario in \cite{guo2018computation} to the energy-constrained Internet of Things (IoT) devices with DC capability to tackle the conflict between resource-hungry mobile applications and energy-constrained IoT devices. Authors in \cite{wu2018green, li2020intelligent}  investigated DC-assisted computation offloading scheme in non-orthogonal multiple access system to minimize the total energy consumption.
Note that these works mainly focused on utilizing DC-assisted mobile edge computing to achieve green-oriented computation offloading for data services. However, the benefits of both edge caching and the complementary sub-6 GHz and mmWave bands to the delay-sensitive VR application in the DC-assisted mobile edge computing scenario have not been investigated.}

\textit{3) Mobile Edge Computing for VR Delivery:}
Authors in \cite{bastug2017toward, mangiante2017vr, elbamby2018toward} inspired the use of MEC resources for VR delivery to obtain  the potential performance gain, but they did not establish theoretical formulation and provided efficient algorithms. Some efficient task offloading algorithms in wireless cellular networks with MEC were proposed in \cite{you2016energy, dinh2017offloading, mao2017stochastic, wang2017computation, bi2020joint}.
The MEC task offloading technology was introduced for the delivery of VR video in \cite{sun2019communications, dang2019joint, chakareski2020viewport}.
Sun \textit{et al.} \cite{sun2019communications} focused on balancing the local and edge computing resources to maximize the average bandwidth usage for VR delivery.
Dang \textit{et al.} \cite{dang2019joint} extended the scenario in \cite{sun2019communications} to an edge fog computing network to achieve an economical computing offloading scheme that minimized the average delay for VR delivery.
Chakareski \textit{et al.} \cite{chakareski2020viewport} proposed an edge server cooperation framework to efficiently offload the local computing tasks.
It is worth noting that \cite{dang2019joint, sun2019communications, chakareski2020viewport} focused on improving the  average  bandwidth usage or delivery latency in MEC empowered VR delivery networks. Nevertheless, little work has focused on improving the reliability of VR delivery.

\subsection{Motivation and Contributions}
As presented in existing related works, wireless channels and connections have a significant impact on the reliability of data transmission. In the context of VR delivery,  these wireless factors have not been well studied.
\textcolor{blue}{
First, the transmission delay and the reliability of wireless VR delivery will be randomly affected by the path loss, fading, and signal blockage of wireless channels.  However, the impact of wireless channel fluctuation on the performance of VR delivery is largely overlooked in the existing literatures.
For example, the authors in \cite{dang2019joint} simplified the characteristics of wireless channels as a deterministic value of the minimum allowable transmission rate of each VR viewpoint, without considering the randomness of wireless channels. To this end, some probabilistic analysis is required to more accurately characterize the impact of wireless channel fluctuation on the transmission delay and the reliability of wireless VR delivery.
}
Second, although the  dual-connectivity \textcolor{blue}{sub-6 GHz} and mmWave technology has been proved to improve the reliability of data transmission, existing works on VR delivery mostly assumed single-connectivity (i.e., one wireless interface) for a HMD (e.g., see \cite{you2016energy, dinh2017offloading, mao2017stochastic, wang2017computation, bi2020joint, dang2019joint, sun2019communications, chakareski2020viewport}). These works mainly focused on the task offloading of VR video projection and rendering in the MEC scenario.  Nevertheless, DC is considered as a promising technology for realizing URLLC applications, \textcolor{blue}{while the resource coordination for VR delivery in DC sub-6 GHz and mmWave HetNets with edge caching and computing capability is challenging when the interface diversity is implemented.}
On the one hand, it is essential to decide whether to deliver the viewpoint over \textcolor{blue}{sub-6 GHz} links or  mmWave links when the requested viewpoint is not cached locally.
This should not only consider the influence of the wireless  channel conditions, but also consider whether the requested viewpoint  is cached at \textcolor{blue}{sub-6 GHz} BSs ($\mu$BSs) or mmWave BSs (mBSs), as well as the impact of  the limited computation resources of HMDs, $\mu$BSs and mBSs on the computation delay.
On the other hand, a proper resource coordination between \textcolor{blue}{sub-6 GHz} networks and mmWave networks is vital for adapting the computation-intensive VR videos.
Specifically, local cache can save the transmission delay of viewpoints at the cost of occupying limited local caching  capacity.
Caching simultaneously at $\mu$BSs and mBSs can reduce the transmission delay of some popular viewpoints with a higher probability,  but at the cost of occupying more caching capacity of BSs  compared with caching at either $\mu$BSs or mBSs.
In addition, the caching and computing strategies in DC \textcolor{blue}{sub-6 GHz} and mmWave HetNets are coupled.
If MVs are cached at $\mu$BSs or mBSs, the caching  capacity can be saved, but the computation time for projection is increased. While caching SVs save the computation time at the cost of more caching  capacity.

To address the aforementioned issues, in this paper, we utilize tools from stochastic geometry to model the DC \textcolor{blue}{sub-6 GHz} and mmWave channels and theoretically analyze the reliability of VR delivery by comprehensively considering the influence of communication, caching, and computation (3C). We then  perform the resource coordination and jointly optimize  the caching and computing strategy, aiming  to ensure higher reliability of VR delivery.
The main contributions of this paper are summarized as follows:
\begin{itemize}
  \item \textcolor{blue}{We present a DC sub-6 GHz and mmWave HetNet architecture empowered by edge caching and computing capability, aiming to meet the requirement of ultra reliability of  VR delivery. Based on this architecture, we conduct a collaborative design of caching and computing resource utilization that adapts to the complementary advantages of sub-6 GHz and mmWave links by utilizing a stochastic geometry analytical framework.}
  \item We derive closed-form expressions for the reliability of VR delivery in the DC \textcolor{blue}{sub-6 GHz} and mmWave HetNet. We propose a link selection strategy based on the  minimum-delay delivery, and derive the VR delivery probability over \textcolor{blue}{sub-6 GHz} links and mmWave links. The relationship between the delivery probability and the difference of caching and computing strategy in \textcolor{blue}{sub-6 GHz} and mmWave tiers is provided. We theoretically demonstrate that \textcolor{blue}{sub-6 GHz} links can be used to enhance the reliability of VR delivery despite the large mmWave bandwidth.
  \item \textcolor{blue}{Based on the analytical work, the coupling 3C resource allocation can be formulated as a joint caching and computing optimization problem, which is designed for the feasibility of practical implementation.}
      By analyzing the coupling relationship between caching and computing strategies at HMDs, $\mu$BSs and mBSs, we transform the problem into a multiple-choice multi-dimension knapsack problem (MMKP). We propose a best-first branch and bound algorithm (BFBB) by calculating the upper and lower bounds of the MMKP to obtain the optimal solution. To further reduce the complex of the algorithm, the problem is transformed into a continuous optimization problem, and the difference of convex programming technique is utilized to obtain a sub-optimal solution.
  \item We conduct numerical evaluation to validate the theoretical analysis of reliability with respect to several key parameters, such as blockage density, CPU cycles, and cache size  of $\mu$BSs/mBSs/HMDs.
  Numerical simulations show great promise of the proposed DC \textcolor{blue}{sub-6 GHz} and mmWave HetNet architecture to enhance the reliability of  VR delivery. The results also reveal that it is more advantageous to cache MVs at $\mu$BSs and cache SVs at mBSs for enhancing the reliability efficiently.

\end{itemize}

The rest of this paper is organized as follows. We first introduce the system model in Section \uppercase\expandafter{\romannumeral2}. We then analyze the reliability of VR delivery over DC \textcolor{blue}{sub-6 GHz} and mmWave HetNets and provide the link selection strategy to formulate the joint caching and computing problem in Section \uppercase\expandafter{\romannumeral3}. In Section \uppercase\expandafter{\romannumeral4}, we first provide the BFBB algorithm based on the computation of upper and  lower bounds of the optimal solution, and a difference of convex programming algorithm with lower complexity is proposed to maximize the reliability of VR delivery. Numerical simulations are provided in Section \uppercase\expandafter{\romannumeral5}, and Section \uppercase\expandafter{\romannumeral6} concludes this paper with summary.

\begin{spacing}{1.42}
\section{System Model}
\subsection{Network Model}
\begin{figure}
\centering
\includegraphics [width = 3.7in]{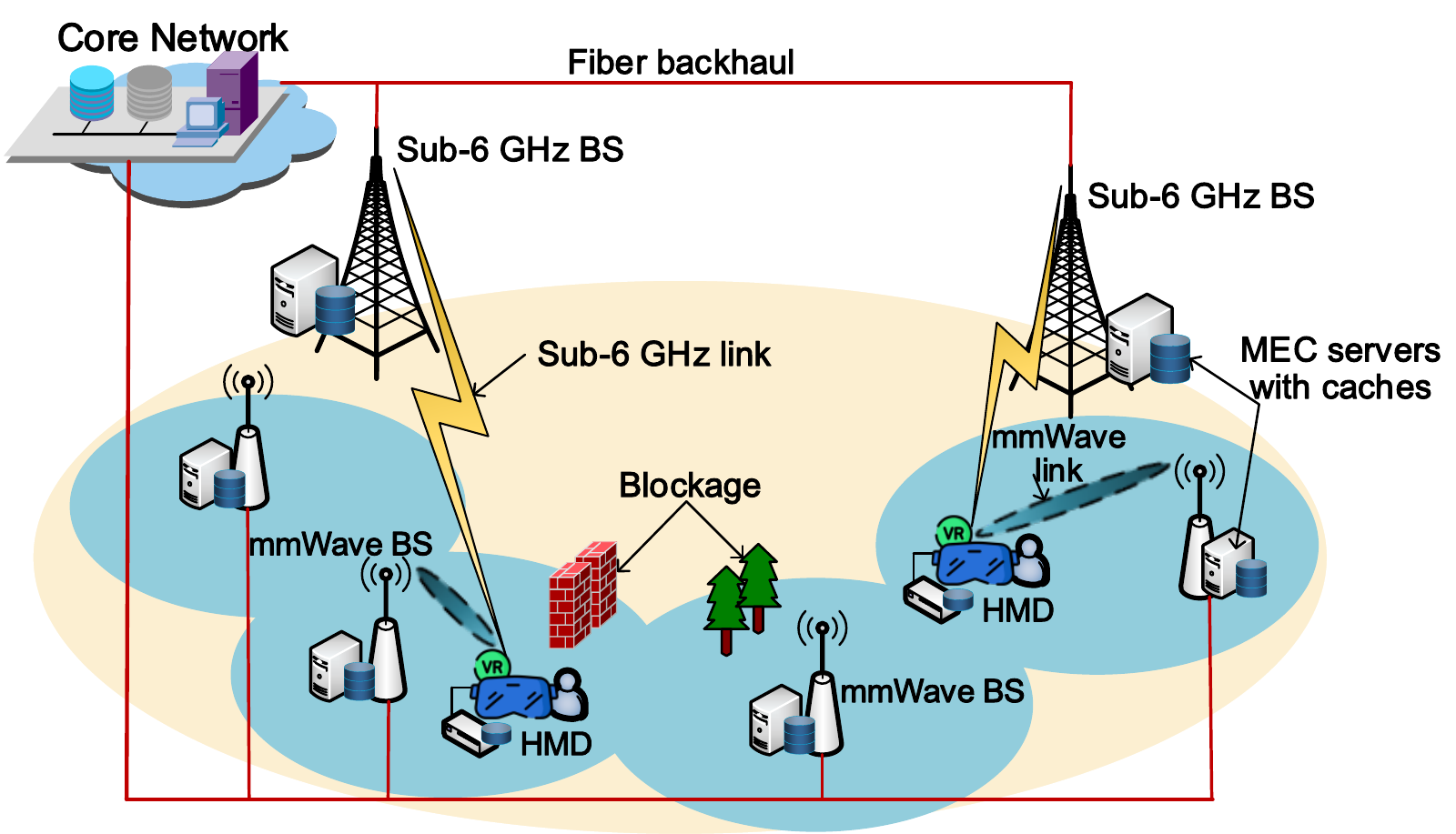}
\caption{Dual-connectivity \textcolor{blue}{sub-6 GHz} and mmWave heterogeneous network architecture empowered by mobile edge capability.}
\label{fig:1}
\end{figure}

As illustrated in Fig. \ref{fig:1}, we consider the VR delivery in a two-tier dual-connectivity (DC) \textcolor{blue}{sub-6 GHz} and mmWave  HetNet.
\textcolor{blue}{Similar to \cite{Yi2018, Biswas2019, Kuang2019}, the locations of \textcolor{blue}{sub-6 GHz} BSs ($\mu$BSs), mmWave BSs (mBSs) and HMDs are modeled as three independent and homogeneous Poisson Point Processes (PPPs)\footnote{\textcolor{blue}{PPPs are generally used for modeling the scenarios that do not consider hot-spot areas (e.g., in rural areas). If considering hot-spot areas, the Poisson cluster process (PCP) can be a more realistic model which captures the clustering nature of base stations and users in hot-spot areas \cite{Chun2015}. Modeling the locations of base stations as PCPs will have an impact on the outage and coverage performance of the HetNet, which will be considered in our future work.}}, which are denoted by $\Phi_{\mu}$, $\Phi_{m}$, $\Phi_{H}$ of densities $\lambda_{\mu}$, $\lambda_{m}$, $\lambda_{H}$, respectively.}
DC is implemented at HMDs by  performing packet data convergence protocol (PDCP) layer integration \cite{semiari2019integrated}, so a HMD can be simultaneously connected to a $\mu$BS and a mBS. Instead of delivering 360$^\circ$ VR videos, $\mu$BSs and mBSs only deliver the requested MVs and SVs. As mentioned above, MVs are  projected and rendered using the computing resources at BSs or  HMDs  to create SVs. The set of all viewpoints is denoted by $\mathcal{J} = \{ 1, 2, \cdots, J \}$. Each viewpoint $j \in \mathcal{J} $ corresponds to an MV and an SV. The data sizes of the $j$-th MV and SV are denoted by $d_j^M$ and $d_j^S$, respectively. Note that $d_j^S$ is at least twice larger than $d_j^M$, i.e., $d_j^S /d_j^M \geq 2$, due to the creation of stereoscopic video.

In the \textcolor{blue}{sub-6 GHz} tier, all channels undergo independent identically distributed (i.i.d.) Rayleigh fading. Without loss of generality, when a  HMD located at the origin $o$ requests the $j$-th viewpoint from the associated $\mu$BS, its received signal-to-interference-plus-noise ratio (SINR) is given by,
\begin{equation}\label{SINR-mu}
  \Upsilon_j^\mu = \frac{P_{\mu} h_j^{\mu} r^{-\alpha_\mu} }{I_j^{\mu} + \sigma_{\mu}^2  },
\end{equation}
where $P_{\mu}$ is the transmit power of $\mu$BS, $h_j^{\mu}$ is the Rayleigh channel gain between the  HMD and its serving $\mu$BS which follows the exponential distribution. $r^{-\alpha_\mu}$ is the path loss with the distance $r$, where $\alpha_\mu$ is the path loss exponent. $I_j^{\mu} = \sum_{n \in \Phi_{\mu} \backslash b_\mu } P_{\mu} h_j^{\mu,n} r^{-\alpha_\mu,n}$ denotes the inter-cell interference, where $b_\mu$ denotes the associated $\mu$BS. $\sigma_{\mu}^2$ is the noise power of a \textcolor{blue}{sub-6 GHz} link.

In the mmWave tier, unlike the conventional \textcolor{blue}{sub-6 GHz} counterpart, mmWave transmissions are highly sensitive to blockage. We adopt a two-state statistical blockage model for each mmWave link as in \cite{bai2014coverage}, such that the probability of the link to be LOS or NLOS is a function of the distance between the  HMD and its serving mBS. Assume that the distance between them is $r$, then the probability that a link of length $r$ is LOS or NLOS can be modeled as
\begin{equation}
    \rho_\mathrm{L}(r) = \mathrm{e}^{-\kappa r},\ \rho_\mathrm{N}(r) = 1 - \mathrm{e}^{-\kappa r},
\end{equation}
respectively, \textcolor{blue}{where $\kappa$ is a parameter determined by the density and the average size of the blockages \cite{bai2014analysis}. Under dual-connectivity mode, we consider the blockage effects for the mmWave tier by using the defined LOS/NLOS probability function.}
Assume that the antenna arrays at mBSs and HMDs perform directional beamforming with the main lobe directed towards the dominant propagation path and having less radiant energy in other directions. For tractability in the analysis, we adopt a sectorial antenna pattern \cite{7105406}. Denote $\theta$ as the main lobe beamwidth, and $M$ and $m$ as the directivity gain of main and side lobes, respectively. Then the random antenna gain/interference $G$ between the mBS and the HMD has 3 patterns with different probabilities, which is given as
\vspace{-5mm}

\begin{small}
\begin{numcases}{G=} \label{random-antenna-gain}
M^2, \quad \text{with prob.} \  (\frac{\theta}{2\pi} )^2, \nonumber  \\
Mm, \quad \text{with prob.} \  \frac{2 \theta (2\pi - \theta)}{(2\pi)^2}, \label{Gx-definition} \\
m^2, \quad \text{with prob.} \  (\frac{2\pi-\theta}{2\pi} )^2. \nonumber
\end{numcases}
\end{small}
Independent Nakagami fading is assumed for each link. Parameters of Nakagami fading $N_\mathrm{L}$ and $N_{\mathrm{N}}$ are assumed for LOS and NLOS links, respectively.
Therefore, when the HMD requests the $j$-th viewpoint from its associated mBS, the received SINR is given by
\begin{equation}\label{SINR-m}
  \Upsilon_j^{m} = \frac{P_{m} h_j^{m} G r^{-\alpha_{m}} }{ I_j^m +  \sigma_{m}^2  },
\end{equation}
where $P_{m}$ is the transmit power of the mBS, $h_j^{m}$ is the Nakagami channel fading which follows Gamma distribution. The path loss exponent $\alpha_{m} = \alpha_{\mathrm{L}}$ when it is a LOS link and $\alpha_{m} = \alpha_{\mathrm{N}}$ when it is an NLOS link. $I_j^m = \sum_{n \in \Phi_{m} \backslash b_m } P_{m} h_j^{m,n} r^{-\alpha_m,n}$, where $b_m$ denotes the associated mBS.  $\sigma^2_{m}$ is the noise power of a mmWave link.
The rate of \textcolor{blue}{sub-6 GHz} and mmWave links are given by the Shannon's formula as $R_j^l = B_l \log_2 (1 + \Upsilon_j^{l}), l \in \{ \mu, m \}$, where $B_l$ denotes the subchannel bandwidth of \textcolor{blue}{sub-6 GHz} or mmWave links.

\subsection{Caching and Computing Model}
\textcolor{blue}{The probability of the service request of the HMD for the $j$-th viewpoint is $p_j, j \in \mathcal{J}$. Considering all $J$ viewpoints, we have $\sum_{j=1}^{J} p_j =1$.
}
$\mu$BSs, mBSs, and HMDs all have caching and computing capabilities. The cache size at $\mu$BSs, mBSs, and HMDs is $C^{\mu}$, $C^{m}$, and $C^{H}$, respectively. The caching decision for the MV and SV of the $j$-th viewpoint is denoted by $y_j^{q,M} \in \{ 0, 1 \}$ and $y_j^{q, S} \in \{ 0,1 \}, q \in \{ \mu,m,H \}$.
$y_j^{q, \omega } = 1, \omega \in \{ M,S \}$ indicates that the MV or SV of the $j$-th viewpoint is cached at device $q$, otherwise $y_j^{q, \omega } = 0$. With the cache size constraint, we have $\sum_{j=1}^{J} d_j^{M} y_j^{q,M} + d_j^S y_j^{q,S} \leq C^{q}$.

The computing decisions at $\mu$BSs, mBSs, and  HMDs are considered to process the projection and rendering. The CPU-cycle frequency of $\mu$BSs, mBSs, and  HMDs is denoted by $f_q, q \in \{ \mu,m,H \}$. The average energy consumption constraint of $\mu$BSs, mBSs, and  HMDs is denoted by $E_q, q \in \{ \mu,m,H \}$.
Define $\varepsilon$ as the number of computation cycles required to process the projection and rendering of one bit input.

\textcolor{blue}{
According to \cite{mao2017survey, yuan2006energy, zhang2013energy}, the energy consumption of a CPU cycle can be expressed as $k_q = \eta_q f_q^2$, where $f_q$ is the CPU-cycle frequency of device $q$, and $\eta_q$ is a constant related to the hardware architecture of device $q$.
}
The computing decision for the $j$-th viewpoint is denoted by $z_j^{q} \in \{ 0, 1 \}$, where $z_j^q = 1$ indicates that the projection is computed at device $q$, otherwise $z_j^q = 0$.
\textcolor{blue}{
Then, the average energy consumption of a device for processing a computing task of a viewpoint can be calculated by averaging all $J$ viewpoints, which can be expressed as $\sum_{j=1}^{J} p_j \varepsilon k_q d_j^M  z_j^q$.
Considering that the HMD is energy-constrained, the average energy consumption of the device for processing a computing task of a viewpoint should be limited, otherwise the onboard battery will be depleted quickly, and the device will heat up due to the overload of the CPU, which will degrade the user experience. Therefore, a constraint of the average energy consumption of the device $\sum_{j=1}^{J} p_j \varepsilon k_q d_j^M  z_j^q \leq E^q$ should be satisfied, which can ensure that the CPU of the device will not be overloaded.
}

In the case where the requested viewpoint cannot be found in HMD cache, $\mu$BS cache or mBS cache, the requested viewpoint is retrieved from the core network through the fiber backhaul links with an extra backhaul retrieving delay $\tau_j^{r}$. We assume that the backhaul transmission rate is $R_j^{b}$, and the backhaul capacity constraint is $B^b$.
\textcolor{blue}{
The notations used in Sec. II to Sec. IV are summarized in Table I.
}

\begin{table*}
  \centering
  \scriptsize
  \textcolor{blue}{\caption{Summary of Notations} }
  \resizebox{\textwidth}{4.1cm}{
  \label{notation-table}
  \begin{tabular}{|m{1.9cm}|m{6.0cm}|m{1.9cm}|m{6.0cm}|}
    \hline
    \textbf{Notation}   & \textbf{Description}        & \textbf{Notation}   & \textbf{Description}  \\\hline
    $\Phi_{\mu} \ / \ \Phi_{m} \ / \  \Phi_{H}$  &   PPP of $\mu$BSs / mBSs / HMDs  &  $\lambda_{\mu} \ / \ \lambda_{m} \ / \ \lambda_{H}$ &  density of $\mu$BSs / mBSs  / HMDs    \\\hline
     $P_{\mu} \ / \  P_{m}$  &   Transmit power of  $\mu$BSs / mBSs      &
     $B_\mu \ / \ B_m$  & Bandwidth for each user at sub-6 GHz / mmWave     \\\hline
     $\alpha_{\mathrm{L}} \ / \ \alpha_{\mathrm{N}}$  & Path loss exponent of LOS and NLOS  &
     $N_\mathrm{L} \ / \ N_\mathrm{N}$   &   Nakagami fading parameter for LOS / NLOS link   \\\hline
     $h_j^\mu \ / \ h_j^m$ &  Channel fading of sub-6 GHz / mmWave links  &
     $M \ / \ m$  & Mainlobe antenna gain / sidelobe antenna gain      \\\hline
     $\Upsilon_j^\mu \ / \ \Upsilon_j^m$   &  Received SINR at the HMD from  $\mu$BSs / mBSs  & $R_j^{\mu} \ / \ R_j^{m}$  &  Data rate of sub-6 GHz / mmWave links     \\\hline
     $\theta$  & Mainlobe beamwidth  &  $\kappa$   &  Blockage density     \\\hline
     $\mathcal{J}$  &  Set of viewpoints   &  $J$  &  The number of viewpoints       \\\hline

     $C^{\mu} \ / \ C^{m} \ / \ C^{H}$   &   The cache size at $\mu$BSs / mBSs /  HMDs    &
     $d_j^M  \ / \ d_j^S$  &  Size of MVs / SVs       \\\hline
     $p_j$  & Request probability of the $j$-th viewpoint  &  $\delta$  &  Skewness of the viewpoint popularity      \\\hline

     $f_\mu \ / \ f_m \ / \ f_H$  &  CPU cycle of  $\mu$BSs / mBSs / HMDs    &
     $\eta_\mu \ / \ \eta_m \ / \ \eta_H$  & Energy efficiency coefficient of HMDs / $\mu$BSs / mBSs   \\\hline
     $T_j$  &  End-to-end delay threshold  of the $j$-th viewpoint    & $\tau_j^\mu \ / \ \tau_j^m$  &  End-to-end latency of the $j$-th viewpoint over the sub-6 GHz / mmWave link  \\\hline
     $A_j^\mu \ / \ A_j^m$  &  Probability of the sub-6 GHz / mmWave link being selected to deliver the $j$-th viewpoint  &
     $\rho_\mathrm{L}(r) \ / \ \rho_\mathrm{N}(r)$  & LOS / NLOS probability of mmWave links with length $r$  \\\hline
     $\varepsilon$  &  The number of computation cycles required for 1 bit input  &      $x_{jk}$  &  Joint caching and computing decision of the $j$-th viewpoint  \\\hline
     $y_j^{q, M} \ / \ y_j^{q, S}$  &  Caching decision for the MV / SV of the $j$-th viewpoint  at device $q$  &
     $z_j^q$  &  Computing decision of the $j$-th viewpoint at device $q$  \\\hline
     $\mathcal{R}_j^{l}$  & Reliability of delivering the $j$-th viewpoint over sub-6 GHz / mmWave links  &
     $\mathcal{R}_j$  &  Reliability of delivering the $j$-th viewpoint in DC sub-6 GHz and mmWave HetNets  \\\hline
     $\tau_j^{l,t} \ / \ \tau_j^{l,c} \ / \ \tau_j^{l,b}$, \ \ $l \in \{ \mu, m \}$   &  Transmission / computing / backhaul retrieving delay  of the $j$-th viewpoint over the sub-6 GHz / mmWave link  &
    $\xi_{jk}^q \ / \ \zeta_{jk}^q \ / \ \varphi_{jk}$  &  Caching occupancy / computing energy consumption / backhaul cost of the $j$-th viewpoint for the $k$-th strategy   \\\hline

  \end{tabular}}
\end{table*}

\section{ Reliability Analysis of VR delivery and Problem Formulation }
In order to investigate the reliability performance of VR delivery in DC \textcolor{blue}{sub-6 GHz} and mmWave  HetNets, we refer to the reliability defined by 3GPP\cite{3GPP}, which is the probability of experiencing an end-to-end latency below the threshold required by the targeted service. Accordingly, using the law of total probability, the reliability of VR delivery can be expressed as,
\begin{equation}\label{reliability-def}
  \mathcal{R} = \sum_{j=1}^{J} p_j \mathbb{P}[\tau_j^{\mu} < T_j \cup \tau_j^{m} < T_j ],
\end{equation}
where $T_j$ is the end-to-end latency threshold\textcolor {blue}{\footnote{\textcolor{blue}{Instead of assuming a pre-determined minimum allowable transmission rate as in \cite{dang2019joint}, we assume a pre-determined end-to-end delay threshold for VR viewpoints, which is a more recognized and practical indicator to ensure the user experience of VR videos, typically no more than 20 ms \cite{han2019mobile}.}}} required by the $j$-th viewpoint, $\tau_j^l, l \in \{ \mu, m \}$ denotes the end-to-end latency when the $j$-th viewpoint is retrieved over a \textcolor{blue}{sub-6 GHz} link or mmWave link. Delay contributions to the end-to-end VR delivery latency include the over-the-air transmission delay $\tau_j^{l,t}$,  the computing delay $\tau_j^{l,c}$, the sensor sampling delay $\tau_j^{s}$, the display refresh delay $\tau_j^{d}$ and the backhaul retrieving delay $\tau_j^{l,b}$ if the requested viewpoint is not cached \cite{elbamby2018toward}. Thus, when the $j$-th viewpoint is delivered over a \textcolor{blue}{sub-6 GHz} link or mmWave link, the end-to-end delay is calculated as
\begin{equation}\label{sum-delay}
  \tau_j^{l} = \tau_j^{l,t} + \tau_j^{l,c} + \tau_j^{l,b} + \tau_j^s + \tau_j^d, j \in \mathcal{J}, l \in \{ \mu, m\}.
\end{equation}
Note that  $\tau_j^{l,t}$ is a random variable which is affected by the channel uncertainty of wireless environments and the data size of the delivered viewpoint. The value of $\tau_j^{l,c}$ can be calculated  when the caching and computing strategy is determined.  Whether the end-to-end delay contains $\tau_j^{l,b}$ depends on whether the requested viewpoint is cached.
$\tau_j^{s}$ and $\tau_j^{d}$ are assumed to be constants.

\subsection{Reliability Analysis over DC \textcolor{blue}{sub-6 GHz} and mmWave HetNets}
The reliability of  VR delivery is mainly affected by the channel uncertainty of wireless environments. In wireless environments where temporary outages are common due to impairments in SINR, VR's non-elastic traffic behavior poses additional difficulty. \textcolor{blue}{In this subsection, we utilize the statistical wireless channel state information for sub-6 GHz links and mmWave links described in Sec. II-A to analyze the reliability of VR delivery.} Utilizing tools from stochastic geometry, tractable expressions can be obtained to characterize the reliability of  VR delivery.

The computing delay $\tau_j^{l,c} = \frac{\varepsilon d_j^M } {f_q}, q \in \{ \mu, m, H \}$ when the $j$-th MV is calculated into SV at device $q$. Considering the caching and computing strategy for the $j$-th viewpoint, the computing delay $\tau_j^{l,c}$ can be expressed as
\begin{align}\label{tau-c}
\tau_j^{l,c} = \frac{\varepsilon d_j^M  (y_j^{H, M} \| y_j^{l, M}) (z_j^H \| z_j^l) } {z_j^{l, M} f_l + (1 - z_j^{l, M}) f_H },
\end{align}
where $\cdot\|\cdot$ is the logical OR operator. And the backhaul retrieving delay $\tau_j^{l,b}$ can be expressed as
\begin{align}\label{tau-b}
\tau_j^{l,b} = (1 - y_j^{l,M} \| y_j^{l, S}) \tau_j^{r}.
\end{align}
\textcolor{blue}{Note that the transmission delay over sub-6 GHz links or mmWave links is affected by the wireless channel fluctuation. Thus, the transmission delay is a random variable related to the distance between communicating nodes, the channel fading and the blockage probability for mmWave links. When the computing delay $\tau_{j}^{l,c}$ and the backhaul retrieving delay $\tau_j^{l,b}$ are obtained under a given caching and computing strategy, we can obtain the transmission delay threshold $T_j^{l,t}$ for the $j$-th viewpoint, which is defined as
\begin{align}
  T_j^{l, t} = T_j - \tau_j^{l,c} - \tau_j^{l,b}.
\end{align}
Then the calculation of the reliability of VR delivery can be transformed into the calculation of the probability that the transmission delay is less than the threshold $T_j^{l,t}$.
}

The data size of the $j$-th delivered viewpoint over wireless links also depends on  the caching and computing strategy, which can be calculated as
\begin{align}
  D_j^{l} = y_j^{l, M} z_j^{l} d_j^S + y_j^{l, M} (1 - z_j^{l}) d_j^M + y_j^{l, S} d_j^S + (1 - y_j^{l,M}) (1 - y_j^{l, S}) d_j^S,
\end{align}
Then, the reliability of delivering the $j$-th viewpoint over  \textcolor{blue}{sub-6 GHz} or mmWave links is a function of $D_j^l$ and $T_j^{l,t}$, which can be expressed as
\begin{align}
  \mathcal{R}_j^{l}(D_j^l, T_j^{l, t}) & = \mathbb{P}[\tau_j^{l} < T_j] =  \mathbb{P}[\tau_j^{l,t } < T_j^{l,t}]
   \overset{(a)}= \mathbb{P}[R_j^{l} > D_j^\omega / T_j^{l,t}]  , l \in \{ \mu, m \}, \label{rel-def}
\end{align}
where (a) follows from $\tau_j^{l,t } = D_j^l/ R_j^{l,t}$,  $D_j^l \in \{ d_j^M, d_j^S \}$.

\begin{proposition}\label{prop1}
When the $j$-th viewpoint is delivered to the HMD over \textcolor{blue}{sub-6 GHz} links, the reliability of VR delivery is given as,
\begin{equation}\label{muwave-reliability}
  \mathcal{R}_j^{\mu}(D_j^\mu, T_j^{\mu, t}) = \sum_{i = 1}^{q} w_{i} \mathrm{e}^{r_{i} + \beta_{\mu}(\nu_j^{\mu}, r_i)}  f_\mu(r_i),
\end{equation}
where $w_{i} = \frac{r_{i}} {(q + 1)^2 [L_{q + 1}(r_{i})]^2 }$, $r_{i}$ is the $i$-th zero of $L_{q}(r)$, $L_{q}(r)$ denotes the Laguerre polynomials, and $q$ is a  parameter balancing the accuracy and complexity. $\nu_j^{\mu} = 2^{\frac{D_j^\mu}{ T_j^{\mu, t} B_{\mu}}} - 1$, $f_\mu(r) = 2 \pi \lambda_\mu r \mathrm{e}^{- \pi \lambda_\mu r^2}$, $\beta_\mu(\nu_j^{\mu}, r) = -\nu_j^{\mu} r^{\alpha_\mu} \sigma_{\mu}^2 - \pi  \lambda_\mu r^2 H_{\delta} (\nu_j^{\mu}) + \pi \lambda_\mu s^{\delta}  \Gamma(1 + \delta) \Gamma(1 - \delta ) $,  $s = \nu_j^{\mu} r^{\alpha_\mu} P_\mu^{-1}$, $\delta = 2 / \alpha_\mu$, $\Gamma(\cdot) $ is the Gamma function, and $H_{\delta}(x) \triangleq {_2}F_{1}(1,\delta; 1+\delta; -x) $ is the Gauss hypergeometric function.
\end{proposition}
\begin{proof}
Please refer to Appendix \ref{AppendixA}.
\end{proof}

The delivery reliability (\ref{muwave-reliability}) is in the form of the complementary cumulative distribution function (CCDF) of SINR over the \textcolor{blue}{sub-6 GHz} tier. The reliability is monotonically decreasing with the SINR threshold $\nu_{j}^{\mu}$, which indicates that the reliability increases with the increase of $T_j^{\mu, t}$, while decreases with the increase of $D_j^{\mu}$.

\begin{proposition}\label{prop2}
When the $j$-th viewpoint  is delivered to the HMD over mmWave links, the reliability of VR delivery is given as,
\begin{equation}\label{mmWave-reliability}
  \mathcal{R}_j^{m}(D_j^m, T_j^{m, t})  =  \sum_{i = 1}^{q} w_{i} \mathrm{e}^{r_{i}}  \beta_{m}(\nu_j^{m}, r_i)    f_{m}(r_i),
\end{equation}
where $\beta_{m}(\nu_j^{m}, r_i) = \sum_{\ell \in \{ \mathrm{L, N}\}} \rho_{\ell}(r_i) \sum_{k=1}^{N_{\ell}}   (-1)^{k+1}      \binom{N_{\ell}}{k}   \mathrm{e}^{ -\frac{k \eta_{\ell} \nu_j^{m} r_{i}^{\alpha_{\ell}}  \sigma_{m}^2 }{P_{m} G}} \mathcal{L}(r_i)$, $\nu_j^{m} = 2^{\frac{D_j^m}{ T_j^{m, t} B_{m}}} - 1$, $f_m(r) = 2 \pi \lambda_m r \mathrm{e}^{- \pi \lambda_m r^2}$, $\eta_\ell = N_\ell (N_\ell !)^{-\frac{1}{N_\ell}}$, and
\begin{small}
\begin{equation} \label{Laplace}
  \mathcal{L}(r) = \prod_{n \in \{ \mathrm{L,N}\}} \prod_{G}  \mathrm{exp} \left[ -2 \pi \lambda_{m} p_G \sum_{u=1}^{N_n} \binom{N_n}{u} \frac{r^{-\frac{1}{\alpha_n} \left(u - \frac{2}{\alpha_n}\right) }} {u \alpha_{n} - 2}  {_2}\textit{F}_{1}  \left(N_n, u - \frac{2}{\alpha_n}; 1 + u - \frac{2}{\alpha_n}; -s   r^{-\frac{1}{\alpha_n}} \right)    \right],
\end{equation}
\end{small}
$\!\!$where $p_G$ is the probability of the random antenna gain defined in (\ref{random-antenna-gain}), and $s = \frac{k \eta_n G \nu_j^{m} r^{\alpha_n}  \sigma_{m}^2} {M^2 N_n}$.
\end{proposition}
\begin{proof}
Please refer to Appendix \ref{AppendixB}.
\end{proof}

Taking into consideration the LOS or NLOS channel state, the delivery reliability (\ref{mmWave-reliability}) is the CCDF of SINR over the mmWave tier. Then some remarks can be concluded from Proposition \ref{prop2}.

\textit{Remark 1}: The reliability of VR delivery over the mmWave tier (\ref{mmWave-reliability}) indicates that the Laplace transform of the interference $\mathcal{L}(r)$ in (\ref{Laplace}) is independent of the transmit power $P_m$.

\textit{Remark 2}: (\ref{mmWave-reliability}) is a monotonically increasing function with respect to $T_j^{m, t}$, and a monotonically decreasing function with respect to $D_j^{m}$.

\textit{Remark 3}: Considering a special case when the viewpoints are only delivered through LOS links, then (\ref{mmWave-reliability})  is a decreasing and convex function with respect to the blockage density $\kappa$.

For DC \textcolor{blue}{sub-6 GHz} and mmWave network, according to the addition rule for probability, the delivery reliability for viewpoint $j$ is calculated as
\begin{align}
  \mathcal{R}_j^{DC} =  & \mathcal{R}_j^{\mu}(D_j^\mu, T_j^{\mu, t}) + \mathcal{R}_j^{m}(D_j^m, T_j^{m, t}) - \mathcal{R}_j^{\mu}(D_j^\mu, T_j^{\mu, t}) \mathcal{R}_j^{m}(D_j^m, T_j^{m, t}),
\end{align}
Note that $\mathcal{R}_j^{DC}$ is the expression of the reliability when the $j$-th viewpoint is delivered through wireless links. On the other hand, if the MV or SV of the $j$-th viewpoint is cached in the HMD, we assume that the end-to-end delay is less than the delay threshold with probability 1, i.e., the reliability is set to 1. Therefore, the complete expression of the reliability can be written as
\begin{align}
\mathcal{R}_j = \textbf{1}(y_j^{H, M} \| y_j^{H, S} = 1 ) + \textbf{1}(y_j^{H, M} + y_j^{H, S} = 0) \mathcal{R}_j^{DC},
\end{align}
where $\textbf{1}(\cdot)$ is the indicator function.

\subsection{Link Selection Strategy and Problem Formulation }

Since the DC mode is adopted at HMD, it is essential to give the criteria for delivery the viewpoints over the \textcolor{blue}{sub-6 GHz} link or the mmWave link. In this paper, since the end-to-end delay is the key factor affecting the reliability of VR delivery, we propose the minimum-delay delivery criteria in DC mode as the link selection strategy, which can be written as
\begin{align} \label{min-delay-criteria}
   & l_0 = \arg \min_{l \in \{ \mu, m \}} \tau_j^l.
\end{align}

\begin{proposition}\label{prop3}
When the $j$-th viewpoint is not locally cached by the HMD, based on the minimum-delay delivery criteria, the probability that the mmWave link is selected to deliver the viewpoint is given as,
\begin{align}
  & A_j^{m} = \int_0^{\infty} p_{{\mathcal{R}_j^\mu}(D_j^\mu, t)} \mathcal{R}_j^{m}(D_j^m, (t - \tau_0)^{+} )  \mathrm{d} t, \label{mmwave-association-prob}
\end{align}
where $p_{{\mathcal{R}_j^\mu}(D_j^\mu, T_j^{\mu,t})}  = \frac{\partial \mathcal{R}_j^{\mu}(D_j^\mu, T_j^{\mu,t}) } {\partial T_j^{\mu,t}} $, $\tau_0 = \tau_j^{m,c} + \tau_j^{m,b} - \tau_j^{\mu,c} - \tau_j^{\mu,b}$, and $(\cdot)^{+} = \max(\cdot, 0)$.
\end{proposition}

\begin{proof}
According to (\ref{min-delay-criteria}), the delivery probability over mmWave links can be expressed as
\begin{align}\label{max-rate-proof}
  A_j^{m} = & \ \mathbb{P} [\tau_j^{m} < \tau_j^{\mu} ]  =  \mathbb{P} [\tau_j^{m,t} < \tau_j^{\mu,t} - \tau_0],
\end{align}
where $\tau_0 = \tau_j^{m,c} + \tau_j^{m,b} - \tau_j^{\mu,c} - \tau_j^{\mu,b}$. Using the results in Propositions 1 and 2, the cumulative distribution function of $\tau_j^{m,t}$ is directly obtained in (\ref{mmWave-reliability}), and denote the probability density function of (\ref{muwave-reliability}) with respect to $T_j^{\mu,t}$ as $p_{{\mathcal{R}_j^\mu}(D_j^\mu, T_j^{\mu,t})}  \triangleq \frac{\partial \mathcal{R}_j^{\mu}(D_j^\mu, T_j^{\mu,t}) } {\partial T_j^{\mu,t}} $, we have
\begin{align}
   A_j^{m} = \int_0^{\infty} p_{{\mathcal{R}_j^\mu}(D_j^\mu, T_j^{\mu,t})} \mathcal{R}_j^{m}(D_j^m, (T_j^{\mu,t} - \tau_0)^{+} )  \mathrm{d} T_j^{\mu,t}.
\end{align}
\end{proof}
Accordingly, the probability that the mmWave link is selected to deliver the viewpoint is given as $A_j^m = 1 - A_j^\mu$. Proposition \ref{prop3} can be used to calculate the expected amount of transmitted data over the \textcolor{blue}{sub-6 GHz} link and the mmWave link. Thus, the expected number of MVs computed at the $\mu$BS, mBS, and the HMD can be obtained, respectively.

Note that $\tau_0$ represents the difference of the delays (except for the transmission delays) between mmWave links and \textcolor{blue}{sub-6 GHz} links, where $\tau_j^{l,c}, \tau_j^{l,b}, l\in\{ \mu, m \}$ can be calculated by (\ref{tau-c}) and (\ref{tau-b}), respectively.
From Remark 2, we can deduce that $A_j^{m}$ is decreasing with $\tau_0$, which indicates that the caching and computing strategy will have an impact on the delivery probability over mmWave and \textcolor{blue}{sub-6 GHz} links. For example, if the blockage density in the mmWave tier is low, we can adjust the caching and computing strategy to get a smaller value of $\tau_0$, thus increasing the delivery probability over mmWave links to improve the reliability of VR delivery.

\textit{Remark 4}: From Remark 3 and the delivery probability over mmWave links (\ref{mmwave-association-prob}), we conclude that when the viewpoints are only delivered through LOS links, $A_j^m$ is a decreasing and convex function with respect to the blockage density $\kappa$. In addition, when $\kappa \rightarrow \infty $, we have $A_j^m \rightarrow 0$, and when $\kappa = 0$, $A_j^m < 1$. This indicates that mmWave links cannot be utilized  when the blockage density is too high, while \textcolor{blue}{sub-6 GHz} links can be utilized even if there is no blockage for mmWave links. In other words, \textcolor{blue}{sub-6 GHz} links can be utilized to enhanced the reliability of VR delivery despite the large mmWave bandwidth.

Next, we can formulate the problem of maximizing the reliability of VR delivery in DC \textcolor{blue}{sub-6 GHz} and mmWave  HetNets by jointly optimizing the caching and computing decision at the $\mu$BS, mBS, and HMD as follow,
\vspace{-8mm}

\begin{small}
\begin{subequations}
\begin{align}
\textbf{P1: } \max_{\{y_j^{q, M}, y_j^{q, S}, z_j^q \}}    &\quad  \sum_{j=1}^{J}  p_j \mathcal{R}_{j}   \label{obj-func} \\
\text{s.t.}    &\quad \sum_{j=1}^{J} d_j^{M} y_j^{q,M} + d_j^S y_j^{q,S} \leq C^{q}, q \in \{ H, \mu, m \} , \label{constraint:cachesize-constraint-original}  \\
&\quad \sum_{j=1}^{J} p_j \varepsilon k_H d_j^M  (y_j^{H,M} + A_j^{\mu, M} y_j^{\mu,M} + A_j^{m,M} y_j^{m,M}) z_j^H \leq E^H  \label{constraint:energy-constraint1} \\
&\quad \sum_{j=1}^{J} p_j \varepsilon A_j^l k_l d_j^M  z_j^l \leq E^l,  l \in \{ \mu, m \}, \label{constraint:energy-constraint2} \\
&\quad  \sum_{j=1}^{J} \sum_{l} A_j^{l} d_j^S (1 - y_j^{H,M} \| y_j^{H,S} \| y_j^{l, M} \| y_j^{l, S} )   \leq B^{b}, j \in \mathcal{J},  l \in \{ \mu, m \}, \label{constraint:backhaul-constraint1}  \\
&\quad  y_j^{q, M} \in \{ 0, 1 \}, y_j^{q, S} \in \{ 0, 1 \}, z_j^{q} \in \{ 0, 1 \},  j \in \mathcal{J} , q \in \{ H, \mu, m \}  \label{constraint:binary-constraint1} ,
\end{align}
\end{subequations}
\end{small}
$\!\!$where (\ref{obj-func}) is the reliability of VR delivery according to (\ref{reliability-def}). (\ref{constraint:cachesize-constraint-original}) is the cache capacity constraint. (\ref{constraint:energy-constraint1}) is the average computing energy consumption at the HMD. (\ref{constraint:energy-constraint2}) is the average computing energy consumption at the $\mu$BS or mBS. (\ref{constraint:backhaul-constraint1}) is the backhaul bandwidth constraint, and (\ref{constraint:binary-constraint1}) ensures the binary decision of the caching and computing decision.
\end{spacing}
\textcolor{blue}{
\textit{Remark 5}: Note that the coupling 3C resources are jointly considered in MEC-enabled networks for performance improvement, though the communication part is not added as an optimization variable in Problem \textbf{P1}. Specifically, the communication strategy $A_j^l$ is affected by the caching and computing strategies reflected by $\tau_0$ as shown in Eq. (\ref{mmwave-association-prob}), while the caching and computing strategies in the formulated Problem \textbf{P1} are affected by the communication strategy $A_j^l$ involved in constraints (\ref{constraint:energy-constraint1})--(\ref{constraint:backhaul-constraint1}).
From the perspective of practical implementation, the caching placement phase is prior to the viewpoint delivery phase (e.g., caching placement is performed during off-peak time), so it is impractical to perform a joint communication, caching and computing optimization. Instead, the  caching and computing strategies are optimized based on the probabilistic communication decision derived in Proposition 3, then the caching and computing decisions can be determined by solving Problem \textbf{P1}, which is a practical solution for 3C resource optimization.
}

\begin{spacing}{1.45}
\subsection{Problem Reformulation}
Problem \textbf{P1} is difficult to solve because the objective function (\ref{obj-func}) and the constraint (\ref{constraint:backhaul-constraint1}) contain logical operations. Fortunately, we can explore the coupling relationship between the caching and computing strategies to greatly reduce the solution space.

\begin{lemma}\label{lemma1}
For any viewpoint $j \in \mathcal{J}$, the following caching and computing relationships hold,
\begin{enumerate}
  \item $(y_j^{H, M} \| y_j^{H,S}) + (y_j^{\mu, M} \| y_j^{\mu, S} \| y_j^{m, M} \| y_j^{m, S}) \leq 1$,
  \item $y_j^{q, M} + y_j^{q, S} \leq 1$, $q \in \{ \mu, m, H \}$,
  \item if $y_j^{l,M} = 1$, then $z_j^l + z_j^H = 1$ holds, $l \in \{ \mu, m \}$.
\end{enumerate}
 \end{lemma}
\begin{proof}
1) means that the HMD and the MEC server (at the $\mu$BS or mBS) should not cache the $j$-th viewpoint simultaneously. This is because when the $j$-th viewpoint is cached at the HMD, the reliability is set to 1, thus there is no need to cache it at the MEC server.

2) means that there is no need to cache the MV and the SV of the $j$-th viewpoint simultaneously at the HMD or the MEC server. This is because the caching and the computing strategy is determined before delivering the $j$-th viewpoint, and caching MV or SV is two different strategies, which should not exist simultaneously.

3) holds because when the MV of the $j$-th viewpoint is cached at the $\mu$BS or mBS, it should be computed either at the MEC server or at the HMD. On the other hand, the computing strategy over the \textcolor{blue}{sub-6 GHz} link can be different from that over the mmWave link.
\end{proof}

According to Lemma 1, the caching and computing strategies are coupled, which can be used to reduce the solution space to 18 joint caching and computing strategies for the $j$-th viewpoint, which are  listed in Table \ref{joint-decision}. In Table \ref{joint-decision}, we rewrite the joint caching and computing strategies of the $j$-th viewpoint
in the form of 9-tuple as $\left[ y_j^{H, M} \ y_j^{H, S} \ z_j^H, y_j^{\mu, M} \ y^{\mu, S} \ z^{\mu}, y^{m, M} \  y^{m, S} \  z^{m}   \right]$. To distinguish whether the MVs computed at the HMD is cached locally, or cached at $\mu$BSs or mBSs, we further rewrite $z_j^H$ as $z_j^H \triangleq (z_j^{H(H)} z_j^{H(\mu)} z_j^{H(m)})$.

The 18 joint caching and computing strategies can be divided into 4 types based on different caching strategies, i.e., local caching, caching at $\mu$BSs and mBSs simultaneously, caching at $\mu$BSs or mBSs, and backhaul retrieving. In each type, different computing strategies are included. Detailed descriptions are as follows,
\begin{itemize}
  \item \textbf{Local caching (Strategy 1, 2):} In this type, the viewpoint $j$ is cached locally at the HMD, i.e. $y_j^{H, M} = 1$, or $y_j^{H, S} = 1$. Since over-the-air transmission and backhaul retrieving are not required in this case, the viewpoint $j$ can be guaranteed to be obtained in time, so the delivery reliability is set to 1. When $y_j^{H, M} = 1$, the $j$-th MV is cached locally, at the cost of $d_j^{M}$ cache size and $k_H d_j^{M} \varepsilon$ computing energy consumption. When $y_j^{H, S} = 1$, the $j$-th SV is directly obtained locally, at the cost of $d_j^S$ cache size.

  \item \textbf{Caching at $\mu$BSs and mBSs simultaneously:} In this type, depending on whether the MV or SV of the $j$-th viewpoint is cached at $\mu$BSs and mBSs, 4 cases can be generated,
      \begin{itemize}
        \item \textbf{(Strategy  3 -- 6)} The $j$-th MV is cached at both $\mu$BSs and mBSs. The delivery reliability over the \textcolor{blue}{sub-6 GHz} link is $\mathcal{R}_j^{\mu}(d_j^S, T_j - \frac{\varepsilon d_j^{M}}{f_{\mu}} )$ when the MV is projected into SV at the $\mu$BS, and $\mathcal{R}_j^{\mu}(d_j^{M}, T_j - \frac{\varepsilon d_j^{M}}{f_{H}} )$ when the MV  is projected into SV at the HMD. Likewise, the delivery reliability over the mmWave link is $\mathcal{R}_j^{m}(d_j^S, T_j - \frac{\varepsilon d_j^{M}}{f_{m}} )$ when the MV is projected into  SV  at the mBS, and $\mathcal{R}_j^{m}(d_j^{M}, T_j - \frac{\varepsilon d_j^{M}}{f_{H}} )$ when the MV is projected into SV
            at the HMD. The computing energy consumption is $A_{q_1} k_{q_2} d_j^{M} \varepsilon, q_1 \in \{ \mu, m \}, q_2 \in \{ \mu, m, H \}$ at the $\mu$BS, mBS or HMD for strategy 3, 4, 5, and $k_{v} d_j^{M} \varepsilon$ at the HMD for strategy 6.

        \item  \textbf{(Strategy  7, 8)} The MV is cached at $\mu$BSs, and the SV
            is cached at mBSs. The delivery reliability over the \textcolor{blue}{sub-6 GHz} link depends on the computing strategy similar to the case in strategy 3 -- 6. The delivery reliability over the mmWave link is $\mathcal{R}_j^{m}(d_j^S, T_j )$. The computing energy consumption is $A_{\mu} k_{q} d_j^{M} \varepsilon, q \in \{ \mu, v \}$ at the $\mu$BS or HMD.

        \item \textbf{(Strategy  9, 10)} The SV is cached at $\mu$BSs, and the MV is cached at mBSs. The delivery reliability over the \textcolor{blue}{sub-6 GHz} link is $\mathcal{R}_j^{\mu}(d_j^S, T_j )$. The delivery reliability over the mmWave link depends on the computing strategy similar to the case in strategy 3 -- 6. The computing energy consumption is $A_{m} k_{q} d_j^{M} \varepsilon, q \in \{ m, v \}$ at the mBS or HMD.

        \item \textbf{(Strategy  11)} The SV is cached at both $\mu$BSs and mBSs. The delivery reliability over the \textcolor{blue}{sub-6 GHz} link is $\mathcal{R}_j^{\mu}(d_j^S, T_j )$, and that over the mmWave link is $\mathcal{R}_j^{m}(d_j^S, T_j )$. The cost of cache size is $d_j^S$ at both $\mu$BSs and mBSs, without computing energy consumption.
      \end{itemize}

  \item \textbf{Caching at $\mu$BSs or mBSs:} In this type, depending on whether the MV or SV is cached at $\mu$BSs or mBSs, 4 cases can be generated,
      \begin{itemize}
        \item \textbf{(Strategy  12, 13)} The MV is cached at $\mu$BSs. The delivery reliability over the \textcolor{blue}{sub-6 GHz} link depends on the computing strategy similar to the case in strategy 3 -- 6.  Considering the backhaul retrieve delay, the delivery reliability over the mmWave link is $\mathcal{R}_j^{m}(d_j^S, T_j - \tau_j^r )$. The computing energy consumption is $A_{\mu} k_{q} d_j^{M} \varepsilon, q \in \{ \mu, v \}$ at the $\mu$BS or HMD, and the backhaul cost is $A_m d_j^S$.

        \item \textbf{(Strategy  14)} The SV is cached at $\mu$BSs. The delivery reliability over the \textcolor{blue}{sub-6 GHz} link is $\mathcal{R}_j^{\mu}(d_j^S, T_j )$, and that over the mmWave link is $\mathcal{R}_j^{m}(d_j^S, \tau_j^r )$. No computing delay or computing energy consumption is generated, and the backhaul cost is $A_m d_j^S$.

        \item \textbf{(Strategy  15, 16)} The MV is cached at mBSs. Considering the backhaul delay, the delivery reliability over the \textcolor{blue}{sub-6 GHz} link is $\mathcal{R}_j^{\mu}(d_j^S, T_j - \tau_j^r )$. The delivery reliability over the mmWave link depends on the computing strategy similar to the case in strategy 3 -- 6. The computing energy consumption is $A_{m} k_{q} d_j^{M} \varepsilon, q \in \{ m, v \}$ at the mBS or HMD, and the backhaul cost is $A_{\mu} d_j^S$.

        \item \textbf{(Strategy  17)} The SV is cached at mBSs. The delivery reliability over the \textcolor{blue}{sub-6 GHz} link is $\mathcal{R}_j^{\mu}(d_j^S, T_j - \tau_j^r )$, and that over the mmWave link is $\mathcal{R}_j^{m}(d_j^S, T_j )$. No computing delay or computing energy consumption is generated, and the backhaul cost is $A_{\mu} d_j^S$.
      \end{itemize}
  \item \textbf{Backhaul Retrieve (Strategy  18):} The delivery reliability over the \textcolor{blue}{sub-6 GHz} link is $\mathcal{R}_j^{\mu}(d_j^S, T_j - \tau_j^r )$, and that over the mmWave link is $\mathcal{R}_j^{m}(d_j^S, T_j - \tau_j^r )$. The viewpoint $j$ is retrieved via the backhaul at the cost of $d_j^S$.
\end{itemize}

\begin{table*}
\footnotesize
  \centering
  \caption{Joint Caching and Computing Strategy in DC sub-6 GHz and mmWave HetNets}
  \resizebox{\textwidth}{4cm}
  {
  \label{joint-decision}
  \renewcommand\arraystretch{1.7} 
  \begin{tabular}{|c|c|c|c|c|c|c|c|c|c|c|c|c|c|}
    \hline
    \multirow{2}*{\textbf{Type}} &  \multirow{2}*{\shortstack{\textbf{Strategy}  \\  \textbf{Index} $k$}}    & \multirow{2}*{\textbf{Joint Decision }}        & \multicolumn{2}{|c|}{\textbf{sub-6 GHz Link}}  & \multicolumn{2}{|c|}{\textbf{mmWave Link}}  &  \multicolumn{3}{|c|}{\textbf{Caching Occupancy $\xi_{jk}^q$}} & \multicolumn{3}{|c|}{\textbf{Computing Energy Consumption $\zeta_{jk}^q$}} & \multirow{2}*{\shortstack{\textbf{Backhaul} \\ \textbf{Cost $\varphi_{jk}$}}}  \\
    \cline{4-13}
    ~  &  ~  &  ~  & $D_j^\mu$ & $\tau_j^{\mu, c} + \tau_j^{\mu,b}$  &  $D_j^m$  & $\tau_j^{m, c} + \tau_j^{m,b}$    &  $\mu$BS  &  mBS  &  HMD  &  $\mu$BS  &  mBS  &  HMD  & ~ \\
    \hline
    \multirow{2}*{\textbf{Local Caching}} & 1  &  [10(100),000,000]  &  --  & -- & -- & -- &   0  &  0  &  $d_j^{M}$  &  0 & 0  &  $k_{H} d_j^{M} \varepsilon$  &  0\\
    \cline{2-14}
    ~ & 2 &  [01(000),000,000]  & -- & -- & -- & --  &  0  &  0  &  $d_j^S$  &  0  &  0  &  0  &  0 \\
   \hline
   \multirow{9}*{\shortstack{  \\ \\ \\ \\ \\ \\ \\ \textbf{Caching at} \\ \textbf{$\mu$BSs and mBSs} \\ \textbf{simultaneously}}}  & 3  &  [00(000),101,101]  & $d_j^S$ &  $\frac{\varepsilon d_j^M}{f_{\mu}}$  &  $d_j^{S}$  &  $\frac{\varepsilon d_j^M}{f_{m}}$  &  $d_j^{M}$   &  $d_j^{M}$  &  0  &  $A_{\mu} k_{\mu} d_j^M \varepsilon$  &  $A_{m} k_{m} d_j^M \varepsilon$  &  0  &  0 \\
    \cline{2-14}
    ~  &4&  [00(001),101,100]  & $d_j^S$  & $\frac{\varepsilon d_j^{M}}{f_{\mu}}$  & $d_j^M$   &  $\frac{\varepsilon d_j^{M}}{f_{H}}$  &  $d_j^{M}$   &  $d_j^{M}$  &  0  &  $A_{\mu} k_{\mu} d_j^{M} \varepsilon$  &  0  &  $A_{m} k_{H} d_j^{M} \varepsilon$  &  0\\
    \cline{2-14}
    ~  &5&  [00(010),100,101]  & $d_j^{M}$  &  $\frac{\varepsilon d_j^{M}}{f_{H}}$   &   $d_j^S$ &  $\frac{\varepsilon d_j^{M}}{f_{m}}$  &  $d_j^{M}$   &  $d_j^{M}$  &  0  &  0  &  $A_{m} k_{m} d_j^{M} \varepsilon$   &  $A_{\mu} k_{H} d_j^{M} \varepsilon$  &  0\\
    \cline{2-14}
    ~  &6&  [00(011),100,100]  & $d_j^{M}$   &  $\frac{\varepsilon d_j^{M}}{f_{H}}$  &  $d_j^{M}$   &   $\frac{\varepsilon d_j^{M}}{f_{H}}$   &  $d_j^{M}$   &  $d_j^{M}$  &  0  &  0  &  0  &  $ k_{H} d_j^{M} \varepsilon$  &  0\\
    \cline{2-14}
    ~  &7&  [00(000),101,010]  &  $d_j^S$ & $\frac{\varepsilon d_j^{M}}{f_{\mu}}$   &  $d_j^S$   & 0  &  $d_j^{M}$   &  $d_j^S$  &  0  &  $ A_{\mu} k_{\mu} d_j^{M} \varepsilon$  &  0  &  0  &  0\\
    \cline{2-14}
    ~  &8&  [00(010),100,010]  &  $d_j^{M}$  & $\frac{\varepsilon d_j^{M}}{f_{H}}$  & $d_j^S$  & 0  &  $d_j^{M}$   &  $d_j^S$  &  0  &  0  &  0  &  $ A_{\mu} k_{H} d_j^{M} \varepsilon$  &  0\\
    \cline{2-14}
    ~  &9&  [00(000),010,101]  & $d_j^S$ &  0  & $d_j^S$  &  $\frac{\varepsilon d_j^{M}}{f_{m}}$   & $d_j^{S}$  &  $d_j^{M}$  &  0  &  0  &  $ A_{m} k_{m} d_j^{M} \varepsilon$&  0  &  0\\
    \cline{2-14}
    ~  &10&  [00(001),010,100]  & $d_j^S$  &  0   &  $d_j^{M}$  &  $\frac{\varepsilon d_j^{M}}{f_{H}}$  &  $d_j^S$   &  $d_j^{M}$  &  0  &  0  &  0  &  $ A_{m} k_{H} d_j^{M} \varepsilon$  &  0\\
    \cline{2-14}
    ~  &11&  [00(000),010,010]  & $d_j^S$ & 0 & $d_j^S$  & 0  &  $d_j^S$   &  $d_j^S$  &  0  &  0  &  0  & 0  &  0\\
   \hline
   \multirow{6}*{\shortstack{ \\ \\ \\ \\ \\ \\ \textbf{Caching at}\\ \textbf{$\mu$BSs or mBSs}}}  &12&  [00(000),101,000]  & $d_j^S$  & $\frac{\varepsilon d_j^{M}}{f_{\mu}}$  & $d_j^S$  &  $\tau_j^r$   &  $d_j^{M}$  &  0  &  0  & $A_{\mu} k_{\mu}  d_j^{M} \varepsilon$  &  0  &  0  &  $A_{m} d_j^S$\\
    \cline{2-14}
    ~  &13&  [00(010),100,000]  & $d_j^{M}$  & $\frac{\varepsilon d_j^{M}}{f_{H}}$  & $d_j^S$  &   $\tau_j^r$   &  $d_j^{M}$  &  0  &  0  &  0  &  0  &  $A_{\mu} k_{H}  d_j^{M} \varepsilon$  &  $A_{m} d_j^S$ \\
    \cline{2-14}
    ~  &14&  [00(000),010,000]  &  $d_j^S$  &  0  &  $d_j^S$  & $\tau_j^r$  &  $d_j^S$  &  0  &  0  &  0  &  0  &  0  &  $A_{m} d_j^S$   \\
    \cline{2-14}
    ~  &15&  [00(001),000,100]  & $d_j^S$ &  $\tau_j^r$  & $d_j^{M}$  & $\frac{\varepsilon d_j^{M}}{f_{H}}$  &  0   &  $d_j^{M}$  &  0  &  0  &  0  &  $A_{m} k_{H}  d_j^{M} \varepsilon$  &  $A_{\mu} d_j^S$\\
    \cline{2-14}
    ~  &16&  [00(000),000,101]  & $d_j^S$  &  $\tau_j^r$  &  $d_j^S$  & $\frac{\varepsilon d_j^{M}}{f_{m}}$  &  0   &  $d_j^{M}$  &  0  &  0  &  $A_{m} k_{m} d_j^{M} \varepsilon$  &  0  &  $A_{\mu} d_j^S$\\
    \cline{2-14}
    ~  &17&  [00(000),000,010]  & $d_j^S$  & $\tau_j^r$  &  $d_j^S$  & 0  &  0   &  $d_j^S$  &  0  &  0  &  0  &  0  &  $A_{\mu} d_j^S$\\
    \hline
   \textbf{\mbox{Backhaul Retrieving}}  &18&  [00(000),000,000]  & $d_j^S$ & $\tau_j^r$  &  $d_j^S$ & $\tau_j^r$  &  0   &  0  &  0  &  0  &  0  & 0  &  $d_j^S$\\
   \hline
  \end{tabular}}
\end{table*}
Denote $k$ as the strategy index of the 18 strategies listed in Table \ref{joint-decision}, and $x_{jk} \in \{ 0,1 \}$ as the binary decision variable for the strategy of the $j$-th viewpoint.  For the $k$-th strategy, the corresponding caching occupancy, computing energy consumption and the backhaul cost is denoted as $\xi_{jk}^q$, $\zeta_{jk}^q, q \in \{ \mu, m, H \}$, and $\varphi_{jk}$, respectively.
Then problem \textbf{P1} can be reformulated as a multiple-choice multi-dimension knapsack problem (MMKP) as follows,
\vspace{-5mm}

\begin{subequations} \small
\begin{align}
\textbf{P2: } \max_{\{x_{jk}\}, j \in \mathcal{J}, k \in \{1,2, \cdots, 18  \}}    &\quad  \sum_{j=1}^{J} \sum_{k=1}^{18} p_j \mathcal{R}_{jk} x_{jk}  \\
\text{s.t.}    &\quad \sum_{j=1}^{J} \sum_{k=1}^{18} \xi_{jk}^{q} x_{jk} \leq C^{q}, q \in \{ H, \mu, m \} , \label{constraint:cachesize-constraint}  \\
&\quad  \sum_{j=1}^{J} \sum_{k=1}^{18} p_j \zeta_{jk}^{q} x_{jk} \leq E^{q}, q \in \{ H, \mu, m \} ,  \label{constraint:energy-constraint}  \\
&\quad  \sum_{j=1}^{J} \sum_{k=1}^{18} \varphi_{jk} x_{jk} \leq B^{b},  \label{constraint:backhaul-constraint}  \\
&\quad  \sum_{k=1}^{18} x_{jk} = 1, j \in \mathcal{J}, \label{constraint:SumOne-constraint} \\
&\quad  x_{jk} \in \{ 0, 1 \} , j \in \mathcal{J} , k \in \{1,2,\cdots, 18\}  \label{constraint:binary-constraint} .
\end{align}
\end{subequations}
Problem \textbf{P2} is a 18-choice 7-dimensional MMKP problem, which is proved to be NP-hard. Specifically, $J$ viewpoints are considered as $J$ classes, while the 18 joint caching and computing strategies belonging to each class $j$ are considered as 18 items. The problem is to choose one item from each class such that the profit sum is maximized while satisfying the capacity and energy consumption constraints.
\end{spacing}

\begin{spacing}{1.38}

\section{Optimization of joint caching and computing  }

\subsection{Optimal solution using BFBB algorithm}
We propose a best-first branch and bound (BFBB) algorithm to solve the MMKP problem \textbf{P2}. The key idea of the BFBB algorithm  can be described as follows:  1) Compute upper bounds and lower bounds of the optimal solution and discard branches that cannot produce a better solution than the best one found so far by the algorithm; 2) Utilize a priority queue to select the node with the highest priority as the expanded node of the search tree.

\textit{1) Computation of upper and lower bound:}
First, to compute the upper bound of problem \textbf{P2}, we construct an auxiliary problem \textbf{P3}, which relaxes the constraints into the sum of constraints (\ref{constraint:cachesize-constraint})--(\ref{constraint:backhaul-constraint}) as follows,
\begin{subequations}
\begin{align}
\textbf{P3: } \max_{\{x_{jk}\}}    &\quad  \sum_{j=1}^{J} \sum_{k \in \mathcal{K}} p_j \mathcal{R}_{jk} x_{jk}    \\
\text{s.t.}    &\quad  \sum_{j=1}^{J} \sum_{k \in \mathcal{K}} \varpi_{jk} x_{jk} \leq C_0 , \label{constraint:cachesize-constraint-relaxed}  \\
 &\quad \text{(\ref{constraint:SumOne-constraint}), (\ref{constraint:binary-constraint})},  \nonumber
\end{align}
\end{subequations}
where $\varpi_{jk} =  \left( \varphi_{jk} + \sum_{q \in \{ H, \mu, m \} } (\xi_{jk}^{q} + p_j \zeta_{jk}^{q}) \right)$ and  $ C_0 =  B^b + \sum_{q \in \{ H, \mu, m \} } \left( C^{q} + E^{q} \right) $.

For each viewpoint $j$, select the joint caching and computing decision $k$ which maximizes $\frac{\mathcal{R}_{jk}} {\varpi_{jk}}$, and denote the corresponding decision as $k_{\mathrm{max}}$ for each viewpoint $j$.
Let $\mathcal{R}_{\mathrm{UB}}$ denote the upper bound of problem \textbf{P3}, and let $\varpi_0 = \sum_{j \in \mathcal{J}} \varpi_{jk_{\mathrm{max}}}$, $\mathcal{R}_0 = \sum_{j \in \mathcal{J}} \mathcal{R}_{jk_{\mathrm{max}}}$, then there exist the following two cases,
\begin{itemize}
  \item $\varpi_0 > C_0$. In this case, the solution of problem $\textbf{P3}$ is upper bounded by
      \begin{equation}\label{UB-expression}
        \mathcal{R}_{\mathrm{UB}} =  \sum_{j \in \mathcal{J}} \mathcal{R}_{jk_{\mathrm{max}}} \times  \left( \frac{C_0} {\sum_{j \in \mathcal{J}} \varpi_{jk_{\mathrm{max}}} } \right).
      \end{equation}
  \item $\varpi_0 < C_0$. In this case, the constraint (\ref{constraint:cachesize-constraint-relaxed}) is not violated when decision $x_{jk_{\mathrm{max}}}$ is selected for all $ j \in \mathcal{J}$. The remaining items are merged into the same class $L$, with items indexed by $\ell = 1,\cdots, N_f$. The remaining items are sorted in descending order of $\frac{\mathcal{R}_{\ell}} {\varpi_{\ell}}$. Then, the problem is converted into choosing the remaining items with the remaining capacity constraint equals to $C_0 - \varpi_0$. Adopting greedy algorithm, we select the items in descending order of $\frac{\mathcal{R}_{\ell}} {\varpi_{\ell}}$ to fill the remaining capacity until the capacity constraint is violated. In particular, denote $ \hat{\ell} \in [ 1, \ell ] $ the index of the item that violates the remaining capacity constraint $C_0 - \varpi_0$, which can be defined as,
      \begin{equation}
        \hat{\ell} = \mathrm{min} \left\{ \hat{\ell} : \sum_{\ell = 1}^{\hat{\ell} - 1 } \varpi_{\ell} \leq C_0 - \varpi_0 < \sum_{\ell = 1}^{\hat{\ell} } \varpi_{\ell} \right\}.
      \end{equation}
    Then the solution of problem $\textbf{P3}$ is upper bounded by
      \begin{equation}
        \mathcal{R}_{\mathrm{UB}} = \mathcal{R}_0 + \sum_{\ell = 1}^{\hat{\ell} - 1} \mathcal{R}_{\ell} + \left( \frac{(C_0 - \varpi_0 ) - \sum_{\ell = 1}^{\hat{\ell} - 1} \varpi_{\ell} }  {\varpi_{q}} \right) \times \mathcal{R}_{\hat{\ell}}.
      \end{equation}

\end{itemize}

\begin{lemma}\label{lemma4}
$\mathcal{R}_{\mathrm{UB}}$ is an upper bound for the auxiliary problem $\textbf{P3}$ as well as the problem $\textbf{P2}$.
\end{lemma}

\begin{proof}
The proof can be shown by contradiction for the first case ($\varpi_0 > C_0$) of this lemma.
Suppose that there exists a solution $\hat{X} = (\hat{x}_{1k_1}, \cdots, \hat{x}_{jk_j}, \cdots, \hat{x}_{Jk_J} ) $ satisfying problem $\textbf{P3}$ with corresponding objective value $\hat{\mathcal{R}} = \sum_{j \in \mathcal{J}} \mathcal{R}_{jk_j}$ such that $\mathcal{R}_0 \times (\frac{C_0} {\varpi_0}) < \hat{\mathcal{R}}$. Then we have $\frac{\mathcal{R}_0} {\varpi_0} < \frac{\hat{\mathcal{R}}} {C_0}$. Since $\hat{\varpi} = \sum_{j \in \mathcal{J}} \varpi_{jk_j} \leq C_0$, we have $\frac{\hat{\mathcal{R}}} {C_0} \leq \frac{\hat{\mathcal{R}}} {\hat{\varpi}}$. Thus $\frac{\mathcal{R}_0} {\varpi_0} < \frac{\hat{\mathcal{R}}} {\hat{\varpi}}$.

According to the descending order of $\frac{\mathcal{R}_{jk}} {\varpi_{jk}}$, the inequality $\frac{\mathcal{R}_{jk_{\mathrm{max}}}} {\varpi_{jk_{\mathrm{max}}}} \geq \frac{\mathcal{R}_{jk_{j}}} {\varpi_{jk_{j}}}, \forall j \in \mathcal{J}$ holds, thus we have
\begin{equation}
   \frac{\sum_{j \in \mathcal{J}}\mathcal{R}_{jk_{\mathrm{max}}}} {\sum_{j \in \mathcal{J}} \varpi_{jk_{\mathrm{max}}}} \geq \frac{\sum_{j \in \mathcal{J}} \mathcal{R}_{jk_{j}}} {\sum_{j \in \mathcal{J}} \varpi_{jk_{j}}},
\end{equation}
which is contradictory to $\frac{\mathcal{R}_{\mathrm{max}}} {\varpi_{\mathrm{max}}} < \frac{\hat{\mathcal{R}}} {\hat{\varpi}}$.

For the second case ($\varpi_0 \leq C_0$), the remaining capacity is filled with the remaining items as a knapsack problem, so the upper bound can be obtained also by the greedy algorithm.

Since problem $\textbf{P3}$ is a relaxed problem of the problem $\textbf{P2}$, it is obvious that $\mathrm{UB}$ is an upper bound for $\textbf{P2}$.

\end{proof}

To determine an initial feasible solution that can be used as the starting lower bound, we develop a heuristic algorithm for obtaining the lower bound of problem $\textbf{P2}$ as described in Algorithm \ref{alg:HLB}. The algorithm is based on the conception of aggregate resource saving \cite{parra2005new} considering the multiple resource constraints in $\textbf{P2}$. Specifically, the solutions for each viewpoint that achieve lowest reliability is selected as initial solution. The solution is then upgraded by choosing a new solution for a viewpoint which has the maximum aggregate resource saving while increasing the objective function of total reliability. The aggregate resource saving is defined as
\begin{equation}
  \Delta \varpi_{jk} = \varpi_{j k_j } - \varpi_{jk}.
\end{equation}

Note that there might be some solutions that violate the constraints but achieve higher reliability. To obtain a tighter lower bound, we downgrade some of the solutions to get a feasible solution, which may increase the total value of reliability.

\renewcommand{\algorithmicrequire}{\textbf{Input:}}
\renewcommand{\algorithmicensure}{\textbf{Output:}}

\begin{algorithm}
\caption{Heuristic Lower Bound Computation Algorithm  }
\label{alg:HLB}
\LinesNumbered
\KwIn {$\mathcal{P}_{\mathrm{s}}$, $F$, $C$, $\epsilon$;}
\KwOut {Optimal solution $\textbf{p}^{*}=(p^*_i)_{i\in \mathcal{F}}$;}
\textbf{initialization:} ;

Select the lowest-valued solution for each viewpoint.

 Replace a strategy of a  viewpoint   which has the highest positive value of $\Delta \varpi_{jk}$ and subject to the resource constraints (\ref{constraint:cachesize-constraint})--(\ref{constraint:backhaul-constraint}) . If no such strategy is found, then a strategy with the highest  $ \Delta \mathcal{R}_{jk} / \Delta \varpi_{jk}$ ;

\eIf { no such solution is found in step 3 }
{
    go to step 8;
}
{
    look for another solution in step 3;
}

\If {there exist higher-valued solutions than the selected solution for any viewpoint }
{
    select a highest-valued strategy $\mathcal{R}_{jk}$ with the minimum aggregate resource consumption $\varpi_{jk}$;
}

Replace a lower-valued strategy $\mathcal{R}_{jk}$ that consumes the maximum aggregate resource $\varpi_{jk}$ with the strategy selected in step 9;

\eIf {a solution is found in \text{step 10}}
{
    \eIf {the solution satisfies the constraints (\ref{constraint:cachesize-constraint})--(\ref{constraint:backhaul-constraint})  }
    {
        look for a better solution in step 3;
    }
    {
        go to step 10 for another downgrade.
    }
}
{
    revive the solution obtained at the end of step 3 and terminate.
}

\end{algorithm}

\textit{2) Best-first principle of priority:}
The candidate nodes to be expanded are stored in the priority queue, and assume that the level of the search tree where a candidate node is located is denoted as $j_0$.  Define the best-first principle of priority as the maximum upper bounds of the candidate nodes, which can be calculated as
\begin{align}\label{best-first-principle}
  \mathcal{P} =  \sum_{j=1}^{j_0}  p_j \mathcal{R}_{jk_j}x_{jk_j} + \widetilde{\mathcal{R}}_{\mathrm{UB}},
\end{align}
where $\widetilde{\mathcal{R}}_{\mathrm{UB}}$ is the upper bound of the remaining capacity constraint problem of a  candidate node as follows,
\begin{subequations}
\begin{align}
\textbf{P4: } \max_{\{x_{jk}\}}    &\quad  \sum_{j=j_0+1}^{J} \sum_{k \in \mathcal{K}} p_j \mathcal{R}_{jk} x_{jk}    \\
\text{s.t.}    &\quad  \sum_{j=j_0 + 1}^{J} \sum_{k \in \mathcal{K}} \varpi_{jk} x_{jk} \leq C_0 - \sum_{j=1}^{j_0} \sum_{k \in \mathcal{K}} \varpi_{jk_j} x_{jk_j}  , \label{constraint:cachesize-constraint-relaxed2}  \\
 &\quad \text{(\ref{constraint:SumOne-constraint}), (\ref{constraint:binary-constraint})}.  \nonumber
\end{align}
\end{subequations}
The best-first principle of priority $\mathcal{P}$ is the sum of the current reliability value of a candidate node and its remaining possible reliability value, i.e., the upper bound of problem \textbf{P4}. Note that the form of problem \textbf{P4} is the same as that of problem \textbf{P3}, thus the upper bound of problem \textbf{P4} can be obtained using the similar method described in Lemma \ref{lemma4}.

The advantage of applying the best-first principle of priority is that the first obtained feasible solution using the branch and bound searching method achieves the maximum value of reliability of VR delivery, which accelerates the searching process to find the optimal solution faster.


\textit{3) The proposed BFBB algorithm:}
In this subsection, we propose the optimal joint caching and computing solution based on BFBB algorithm to solve problem \textbf{P2}, as described in Algorithm \ref{alg:BB}.
The BFBB algorithm starts by initializing 18 strategies of the first  viewpoint into the priority queue, and using Algorithm \ref{alg:HLB} to obtain the lower bound of the solution. The next node to be expanded is selected according to the best-first principle of priority (\ref{best-first-principle}), and the lower bound is then updated by directly computing the solution of $\textbf{x}_{j_0 k_{j_0}}$. To perform the branch and bound operation, the feasibility of the expanded child node is verified, and the upper bound of the extended child node is computed using (\ref{best-first-principle}). To accelerate the searching process, the lower bound $\mathcal{R}_{\mathrm{LB}}$ is updated using Algorithm \ref{alg:HLB} in step \ref{step13}, thus more  branches of the search tree can be pruned in step \ref{step11}.
The optimality of the BFBB algorithm is guaranteed by the best-first principle. Since the nodes added into the PQ all correspond to feasible solutions, when $j_0 = J$, i.e., at step 5, the algorithm obtains the biggest value such that $\textbf{x}_{Jk_J}$ is feasible. Then $\textbf{x}_{Jk_J}$ is the optimal solution for problem \textbf{P2}.

\begin{algorithm}
\caption{The Optimal Joint Caching and Computing Solution based on BFBB }
\label{alg:BB}
\LinesNumbered
\KwIn {$p_j, \mathcal{R}_{jk}, \xi_{jk}^{q}, \zeta_{jk}^{q}, \varphi_{jk}^{q}, C^q, E^q, B^b, j \in \mathcal{J}, k \in \mathcal{K}, q \in \{ H, \mu, m \}$;}
\KwOut {Optimal solution $\{x_{jk}^{*}\}, j \in \mathcal{J}, k \in \mathcal{K}$;}
\textbf{initialization:} Using Algorithm \ref{alg:HLB} to obtain the lower bound of the solution $\mathcal{R}_{\mathrm{LB}}$; denote the priority queue as $\text{PQ}$;

\While{\rm{PQ} $\neq \varnothing$}
{
    Select a node $\textbf{x}_{j_0 k_{j_0}} \triangleq \{ x_{1 k_1}, x_{2 k_2}, \cdots, x_{j_0 k_{j_0}} \} $ from \rm{PQ} according to the best-first principle of priority (\ref{best-first-principle}), and remove $\textbf{x}_{j_0 k_{j_0}}$ from \rm{PQ};

    Update the lower bound of the solution $\mathcal{R}_{\mathrm{LB}}$ using  $\textbf{x}_{j_0 k_{j_0}}$;

     \If{$j_0 = J $ }
    {
        Exit the loop with the optimal solution $\textbf{x}_{J k_{J}}$;
    }

    Expand the node $\textbf{x}_{j_0 k_{j_0}}$ to obtain its child nodes $\mathbb{C}$;

    \While{ $\mathbb{C} \neq \varnothing$  }
    {
        \If{ $\textbf{x}_{{(j_0+1)} k}$ corresponds to a feasible solution }
        {
            Calculate the upper bound of the solution $\mathcal{R}_{\mathrm{UB}}$ containing $\textbf{x}_{{(j_0+1)} k}$ using (\ref{best-first-principle});

            \If{$\mathcal{R}_{\mathrm{UB}} > \mathcal{R}_{\mathrm{LB}}$   \label{step11} }
            {
                add $\textbf{x}_{{(j_0+1)} k}$ into \rm{PQ};

                Update the current maximum value of $\mathcal{R}_{\mathrm{LB}}$ using Algorithm \ref{alg:HLB}; \label{step13}
            }

        }

    }

}
\end{algorithm}

\subsection{Suboptimal solution based on DCP}
In this subsection, a suboptimal solution for solving \textbf{P2} is proposed by using difference of convex programming (DCP) with lower computation complexity. First, we equivalently transform \textbf{P2} into a continuous optimization problem, which can be written as,
\begin{subequations}
\begin{align}
\textbf{P5: } \min_{\{x_{jk}\}}   &\quad  \sum_{j \in \mathcal{J}} \sum_{k \in \mathcal{K}} - p_j \mathcal{R}_{jk} x_{jk}   \\
\text{s.t.} &\quad  0 \leq x_{jk} \leq 1, j \in \mathcal{J} , k \in \mathcal{K}, \label{constraint:continuous-variable} \\
&\quad  \sum_{j \in \mathcal{J}} \sum_{k \in \mathcal{K}}  x_{jk} (1 - x_{jk}) \leq 0, \label{constraint:concave-constraint}  \\
&\quad   (\ref{constraint:cachesize-constraint}),  (\ref{constraint:energy-constraint}), (\ref{constraint:backhaul-constraint}), (\ref{constraint:SumOne-constraint}) , \nonumber
\end{align}
\end{subequations}
where the binary variable constraint (\ref{constraint:binary-constraint}) in \textbf{P2} is substituted by continuous variable constraints (\ref{constraint:continuous-variable}) and (\ref{constraint:concave-constraint}). Thus, by solving the continuous optimization problem \textbf{P5}, the computation complexity is greatly reduced compared with that by directly solving \textbf{P2}. Nevertheless, the constraint (\ref{constraint:concave-constraint}) is a concave function instead of a convex function, which becomes the main difficulty for solving \textbf{P5}.

\begin{algorithm}
\caption{Suboptimal solution based on DCP  }
\LinesNumbered
\label{alg:DCA}
\KwIn {$p_j, \mathcal{R}_{jk}, \xi_{jk}^{q}, \zeta_{jk}^{q}, \varphi_{jk}^{q}, C^q, E^q, B^b, j \in \mathcal{J}, k \in \mathcal{K}, q \in \{ H, \mu, m \}$;}
\KwOut {Suboptimal solution of $\{x_{jk}  \}$;}

\textbf{Initialize:} counter $i = 0$, obtain an initial feasible solution $\{ x_{jk}^{0} \}$ using Algorithm \ref{alg:HLB}, threshold $\epsilon = 10^{-5}$;

\Repeat
{$\left( f_0(\textbf{x}^{i-1}) - \phi f(\textbf{x}^{i-1}) \right) - \left( f_0(\textbf{x}^i) - \phi f(\textbf{x}^i) \right) \leq \epsilon $}
{
\textbf{Compute:} $\hat{f}(\textbf{x}; \textbf{x}^{i}) \triangleq f(\textbf{x}^{i}) + \nabla f(\textbf{x}^{i})^{T} (\textbf{x} - \textbf{x}^{i})$.

\textbf{Solve:} Set the value of $\{ x_{jk}^{i+1} \}$ to be a solution of the following convex problem:
\begin{align}
\min_{\{x_{jk}\}}  &\quad f_0 (\{ x_{jk} \}) - \phi \hat{f}(\{ x_{jk} \}; \{ x_{jk}^{i} \}) \nonumber \\
\text{s.t.} &\quad   (\ref{constraint:cachesize-constraint}),  (\ref{constraint:energy-constraint}), (\ref{constraint:backhaul-constraint}), (\ref{constraint:SumOne-constraint}), (\ref{constraint:continuous-variable}). \nonumber
\end{align}

where $f_0 (\{ x_{jk} \}) = \sum_{j \in \mathcal{J}} \sum_{k \in \mathcal{K}} - p_j \mathcal{R}_{jk} x_{jk}$.

\textbf{Update:} $i \leftarrow i + 1$, $\{x_{jk}^{i} \} \leftarrow  \{ x_{jk}^{i-1} \}$;
}

\end{algorithm}

To tackle this difficulty, we utilize a penalty function to bring the concave constraint into the objective function, which can be written as,
\begin{subequations}
\begin{align}
\textbf{P6: } \min_{\{x_{jk}\}}   &\quad  \sum_{j \in \mathcal{J}} \sum_{k \in \mathcal{K}} - p_j \mathcal{R}_{jk} x_{jk} - \phi f(\textbf{x})  \\
\text{s.t.} &\quad   (\ref{constraint:cachesize-constraint}),  (\ref{constraint:energy-constraint}), (\ref{constraint:backhaul-constraint}), (\ref{constraint:SumOne-constraint}), (\ref{constraint:continuous-variable}), \nonumber
\end{align}
\end{subequations}
where $\phi$ is a penalty parameter and $\phi > 0 $, $f(\textbf{x}) = \sum_{j \in \mathcal{J}} \sum_{k \in \mathcal{K}} x_{jk} (x_{jk} - 1) $. To ensure the equivalence between \textbf{P5} and \textbf{P6}, according to the exact penalty property of DCP in \cite{Thi2012exact}, $\phi $ should satisfy,
\begin{equation}\label{penalty-condition}
  \phi > \frac{- p_j \mathcal{R}_{jk} x_{jk}^{0} - g(0)} {\phi_0},
\end{equation}
where $x_{jk}^{0}$ denotes any feasible solution, $g(\phi)$ denotes the optimal objective value for \textbf{P6}, and $\phi_0 = \min_{\textbf{x}} \{ \sum_{j \in \mathcal{J}} \sum_{k \in \mathcal{K}} x_{jk} (1 - x_{jk} ) : (\ref{constraint:cachesize-constraint}),  (\ref{constraint:energy-constraint}), (\ref{constraint:backhaul-constraint}), (\ref{constraint:SumOne-constraint}), (\ref{constraint:continuous-variable})\}$.

Note that the objective function in \textbf{P6} is in the form of a difference of two convex functions. Therefore, we adopt DCP to solve this problem as outlined in Algorithm \ref{alg:DCA}. Specifically, the objective function is convexified by the affine minorization function of the second term $f(\textbf{x})$ in step 3.
The complexity of Algorithm \ref{alg:DCA} lies in solving a series of convex problems in steps 3-6. These convex problems can be solved using the polynomial interior point method \cite{boyd2004convex}, which requires a third degree polynomial complexity in terms of the number of variables. Let $L$ be the total number of iterations, then the complexity of Algorithm \ref{alg:DCA} is  $\mathcal{O}(L (JK)^3)$, where $J$ is the number of viewpoints, and $K$ is the number of strategies for the $j$-th viewpoint.

\textcolor{blue}{
\textit{Remark 6: }
It is worth mentioning that the authors in \cite{dang2019joint} proposed to use the Lagrangian dual decomposition (LDD) approach to obtain a suboptimal solution of the MMKP problem. The Lagrangian relaxation method was adopted in \cite{dang2019joint} to relax all the resource constraints into the objective function, without considering the specific structure of the objective function and the constraints, so it is hard to obtain a high-quality solution close to the optimal solution by using the LDD approach. Besides, the solution of the Lagrangian relaxed problem is not necessarily a feasible solution to the original problem, so the initial condition need to be changed multiple times to obtain a feasible solution of the original problem.
In comparison, we first equivalently transform the MMKP problem into a continuous non-convex optimization problem. To deal with the concave function in the constraints, we exploit the specific structure of the problem and bring the concave constraint into the objective function. Note that this is also an equivalent problem transformation owing to the use of exact penalty property of the DCP problem. Further, the DCP problem is solved by approximating the penalty function as its affine minorization function, which can obtain a high-quality feasible solution effectively, and will be validated in Sec. V.
}

\textcolor{blue}{
\textit{Remark 7 (Practical Implication):} The proposed BFBB and DCP algorithms are implemented based on the pre-known statistical channel state information of sub-6 GHz and mmWave links, as well as the pre-known caching, computing and bandwidth resource information, thus they are offline algorithms which can obtain the optimal and sub-optimal joint caching and computing strategy, respectively. The overall decision is get made in two steps in the practical implementation. First, during the off-peak time, the HMD and the MEC servers obtain the joint caching and computing strategy by using the proposed BFBB or DCP algorithm, and perform caching placement according to the caching strategy. Second, in the viewpoint delivery phase, the HMD calculates the end-to-end delay of sub-6 GHz links and mmWave links based on the channel state information and the joint caching and computing strategy. Then the HMD selects the link to receive viewpoint data according to the minimum-delay delivery criteria, thus the communication strategy can be determined.
}

%
%

\vspace{3mm}

\section{Performance Evaluation}
In this section, the performances of the proposed BFBB and DCP algorithms are evaluated through numerical simulations. The default parameter values of network environment, mobile edge resources and VR videos are listed in Table \ref{simulationsetting}.
\textcolor{blue}{
The $\mu$BSs, mBSs, and HMDs are scattered in a square area of 1km $\times$ 1km based on three independent PPPs with densities listed in Table \ref{simulationsetting}, and the mobility of the HMDs follows the random-walk model \cite{chiang20042}.
}
\textcolor{blue}{For the sake of comparison, five benchmark algorithms are provided including the LDD algorithm used in \cite{dang2019joint}, and the other four benchmark algorithms listed as follows,
}
\begin{table}
  \centering
  \caption{Parameter Values}
  \resizebox{\textwidth}{4cm}{
  \label{simulationsetting}

  \begin{tabular}{|l|l|l|}
    \hline
    \textbf{Parameters}   & \textbf{Physical meaning}        & \textbf{Values}\\\hline
     $P_{\mu}$ /  $P_{m}$  &   Transmit power of  $\mu$BSs / mBSs      &30  / 30 dBm      \\\hline
     $B_\mu$ / $B_m$  & Bandwidth assigned to each user at sub-6 GHz / mmWave  &  100  / 500 MHz  \\\hline
     $\alpha_{\mathrm{L}}$ / $\alpha_{\mathrm{N}}$  & Path loss exponent of LOS and NLOS  & 2.5 / 4     \\\hline
     $\theta$  & Mainlobe beamwidth  &  30$^\circ$      \\\hline
     $M$ / $m$  & Mainlobe antenna gain / sidelobe antenna gain & 10  / -10 dB  \\\hline
     $\kappa$   &  Blockage density  &  6$\times$10$^{-4}$  (Unless otherwise stated)  \\\hline
     $N_\mathrm{L}$ / $N_\mathrm{N}$   &   Nakagami fading parameter for LOS and NLOS channel & 3 / 2  \\\hline
     $\lambda_{\mu}$ / $\lambda_{m}$  \textcolor{blue}{/ $\lambda_{H}$} &  density of $\mu$BSs / mBSs  \textcolor{blue}{/ HMDs}  &  10$^{-5}$ / 3$\times$10$^{-5}$  \textcolor{blue}{/ 10$^{-4}$}   nodes/m$^2$   \\\hline
     $T_j$  &  End-to-end delay threshold  &  20 ms (Unless otherwise stated)   \\\hline
     $\delta$  &  Skewness of the viewpoint popularity   &  0.8  \\\hline
     \textcolor{blue}{$C^{\mu}$  /  $C^{m}$  /  $C^{H}$ }  &   \textcolor{blue}{The cache size at $\mu$BSs / mBSs /  HMDs }   &  \textcolor{blue}{30 / 30 / 10 Mb } \\\hline
     $J$  &  The number of viewpoints  &  20  \\\hline
     $d_j^M$ / $d_j^S$  &  Size of MVs / SVs  &   1  / 3 Mb    \\\hline
     $\varepsilon$  &  The number of computation cycles required for 1 bit input  &  10    \\\hline
     $f_H$ / $f_\mu$ / $f_m$  &  CPU cycle of HMDs / $\mu$BSs / mBSs    &   10$^9$ / 3$\times$10$^9$ / 3$\times$10$^9$   \\\hline
     $\eta_H$ / $\eta_\mu$ / $\eta_m$  & Energy efficiency coefficient of HMDs / $\mu$BSs / mBSs   &  10$^{-25}$ / 10$^{-26}$ / 10$^{-26}$   \\\hline
     $\tau_j^r$  & Backhaul retrieving delay   &  10 ms  \\\hline
  \end{tabular}}
\end{table}

\begin{itemize}
  \item Local SV caching (LSC):
    First, the SVs are cached locally at the HMD according to the descending order of the popularity of the viewpoints until the caching capacity is full. Then, the remaining SVs  are cached at $\mu$BSs and mBSs according to the descending order of the popularity until the caching capacities are full. Then the remaining viewpoints are delivered via the backhaul.

  \item Local computing with MV caching preferentially (LC-MCP):
  The MVs are cached at the HMD according to the descending order of the popularity of the viewpoints until the local computing resources are fully utilized. If there remains caching capacity, the SVs are cached until the caching capacity is full. The remaining  MVs  are  cached at $\mu$BSs and mBSs until the edge computing resources or caching capacities are fully utilized.

  \item $\mu$BS-only: Only $\mu$BSs are deployed in the VR delivery network, i.e., let $\lambda_{m} = 0$, and use the proposed BFBB algorithm to make the joint caching and computing decision.

  \item mBS-only: Only mBSs are deployed in the VR delivery network, i.e., let $\lambda_{\mu} = 0$, and use the proposed BFBB algorithm to make the joint caching and computing decision.
\end{itemize}


\subsection{Reliability Performance of the Proposed Algorithms}
\textcolor{blue}{We first evaluate the reliability of VR delivery with various network parameters as shown in Fig. \ref{fig:relia-blockage} and Fig. \ref{fig:relia-threshold}. It is observed that the BFBB algorithm achieves the highest reliability, and the DCP algorithm achieves reliability close to that of the BFBB algorithm, which is 97.8\% of that of the BFBB algorithm on average. This is because the proposed algorithms comprehensively consider the characteristics of the \textcolor{blue}{sub-6 GHz} link and the mmWave link to determine a joint caching and computing strategy. Note that the LDD algorithm is not able to achieve the same reliability performance as the DCP algorithm, which is mainly due to the operation of relaxing all the constraints into the objective function without exploiting the specific structure of the MMKP problem as the DCP algorithm does.}

\textcolor{blue}{
The reliability of VR delivery with various blockage densities is shown in Fig. \ref{fig:relia-blockage}.
}
\textcolor{blue}{
The blockage model is in accordance with the NLOS probability function defined in Sec. II-A, and the blockages are generated according to the blockage density parameter $\kappa$ and the NLOS probability function. We assume independent LOS probability between different links, i.e., potential correlations of blockage effects between links are ignored, which causes negligible loss of accuracy for performance evaluation \cite{bai2014coverage, bai2014analysis}.
}
It is observed that the performance of LC-MCP is better than that of LSC. This is because caching MVs in the HMD can save more caching capacity and more viewpoints can be cached. Using the local computing capability of the HMD can help more viewpoints to be played in time. The mBS-only can achieve performance close to the proposed algorithms when the blockage density is low, while the performance is significantly reduced when the blockage density is increased. This is because the blockage effect in the mmWave tier causes the link rate to drop rapidly, and the viewpoints cannot be transmitted through the relatively stable \textcolor{blue}{sub-6 GHz} link, making the reliability of VR delivery very low. In contrast, 
\textcolor{blue}{the reliability of VR delivery using the $\mu$BS-only strategy remains unchanged with various blockage densities. This is because the wireless channel fluctuation of the sub-6 GHz tier is not affected by the blockages, which makes the reliability independent of blockage factors.}

\begin{figure}
  \centering
  \subfigure[]{
    \label{fig:relia-blockage} 
    \includegraphics[width=0.31\textwidth]{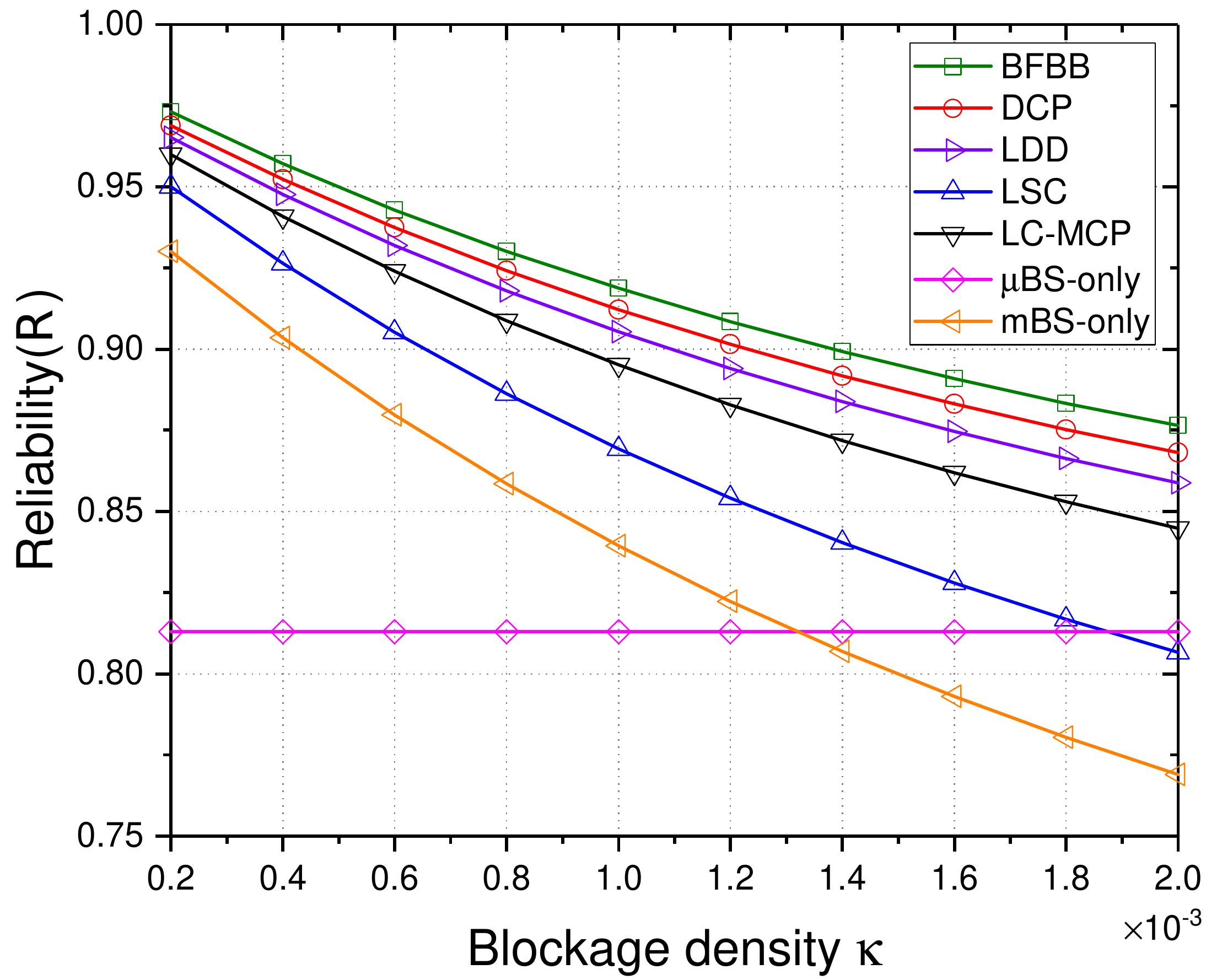}}
  \hspace{0.02in}
  \subfigure[]{
    \label{fig:relia-threshold} 
    \includegraphics[width=0.31\textwidth]{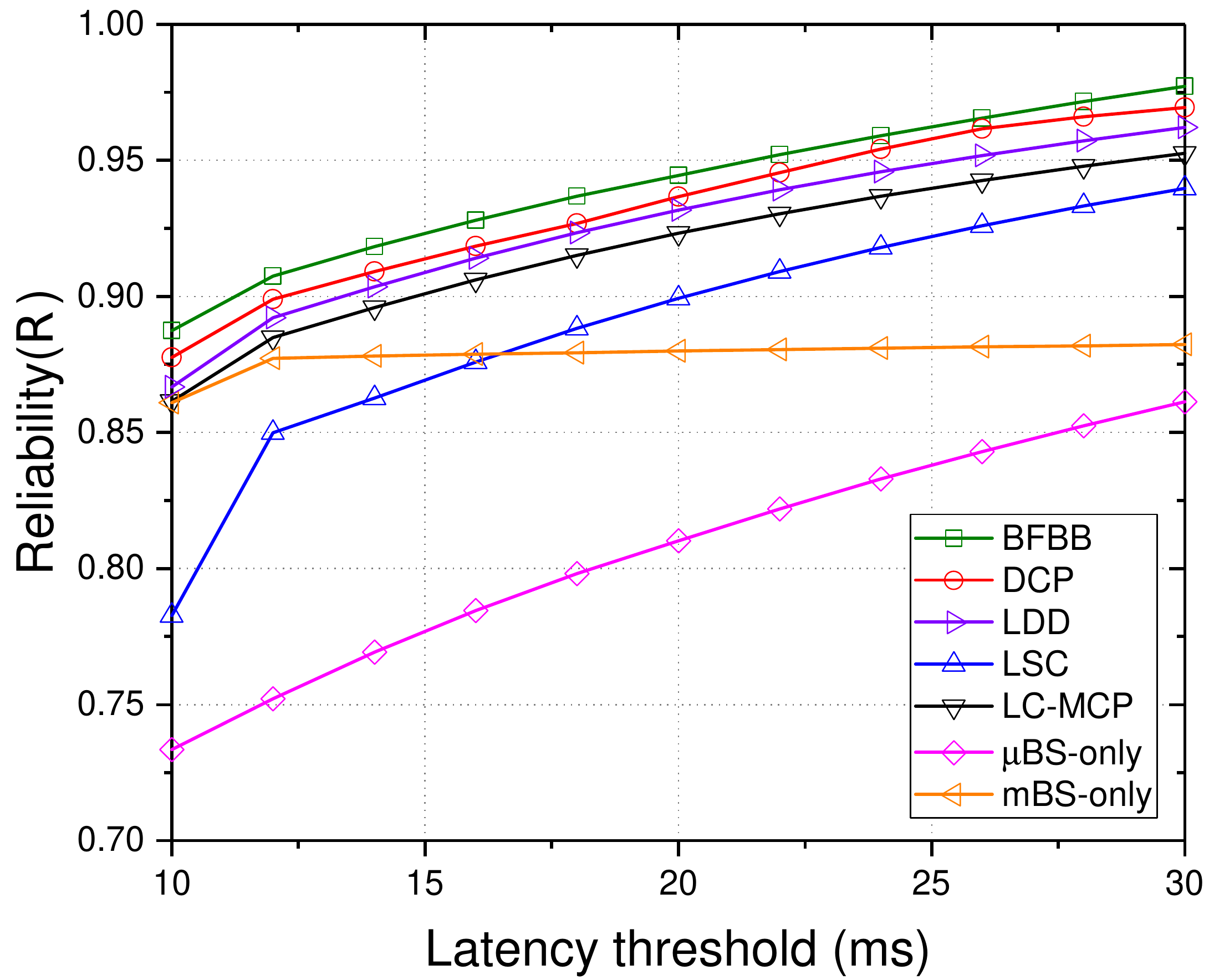}}
  \hspace{0.02in}
  \subfigure[]{
    \label{fig:RunTime} 
    \includegraphics[width=0.303\textwidth]{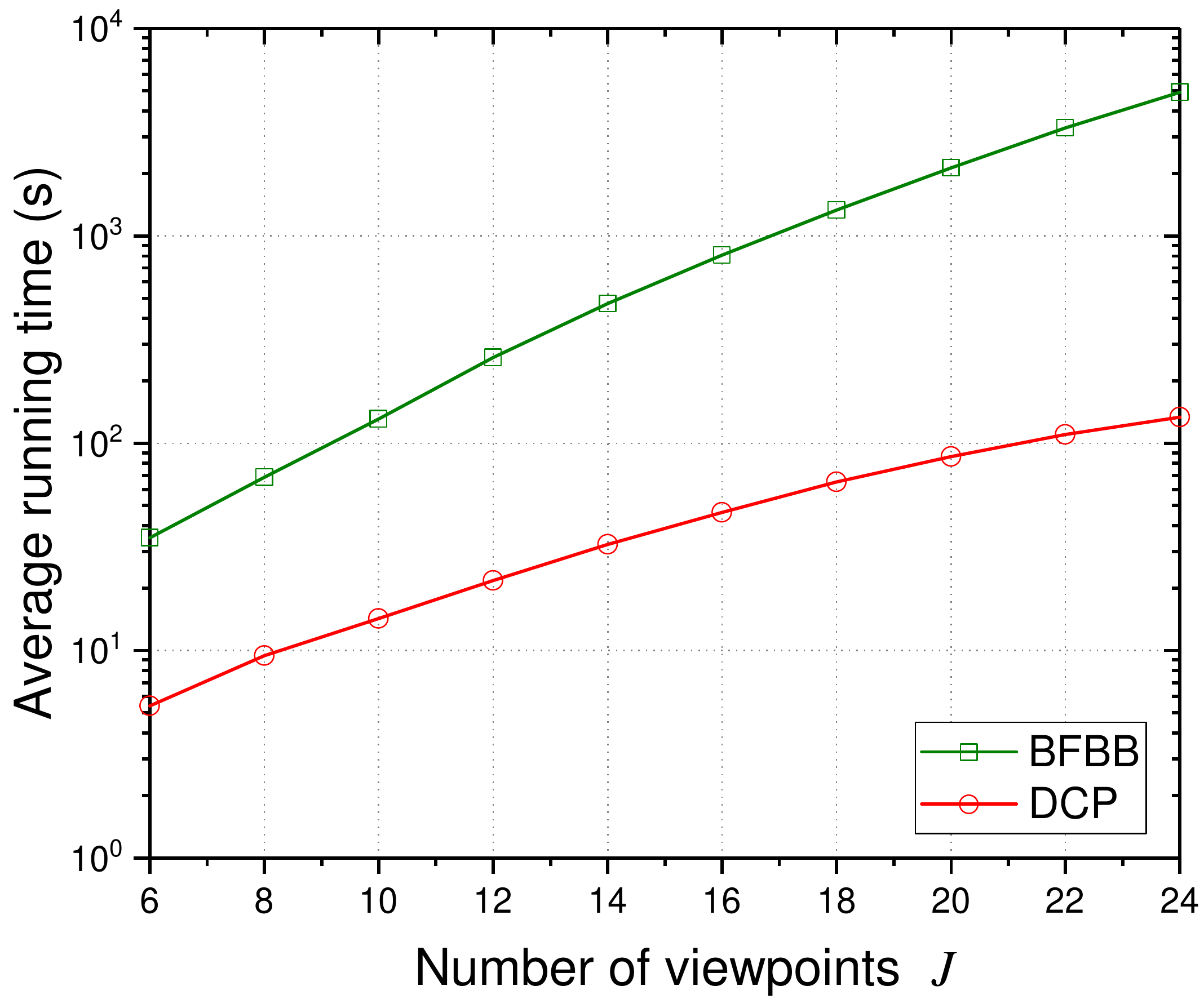}}
  \caption{\textcolor{blue}{(a) Reliability of VR delivery versus blockage density $\kappa$, and (b) reliability of VR delivery versus latency threshold, (c) average running time of the proposed algorithms with various numbers of viewpoints.} }
\end{figure}


Fig. \ref{fig:relia-threshold} shows the reliability of VR delivery with various  delay thresholds. It can be seen that the proposed BFBB and DCP  algorithms also achieve higher reliability. The performance of LC-MCP is better than that of LSC, especially when the delay threshold is low. This is because LC-MCP can cache more viewpoints locally or at BSs, eliminating the backhaul retrieve delay, which is more important when the delay threshold is low. Observing the $\mu$BS-only and mBS-only algorithms, it is seen that mBS-only basically does not change with the increase of the delay threshold, while  $\mu$BS-only  increases much with the increase  of the delay threshold. This indicates that the CDFs of the coverage probabilities of the two channels have different characteristics, and it is necessary to combine the complementary advantages of the two channels to achieve higher reliability.

\textcolor{blue}{
Fig. \ref{fig:RunTime} shows the average running time of the proposed algorithms with various numbers of viewpoints, which is tested on the computer with an Intel Core CPU i5 with a clock rate of 2.30 GHz, and an RAM with 8 G. It is observed that the DCP algorithm runs much faster than the BFBB algorithm, especially when the number of viewpoint becomes larger. This indicates that the DCP algorithm is a more practical algorithm than the BFBB algorithm when applied to the scenario where the number of viewpoints is large.
On the other hand, the BFBB algorithm is also meaningful, because it can be used to obtain the optimal solution of the MMKP problem, which can be viewed as the upper bound of the reliability performance in the considered scenario. Based on this optimal solution, the distance between the results obtained by the suboptimal algorithm and the optimal algorithm can be calculated, which can be used to quantitatively measure the performance of the suboptimal algorithm.
}

\begin{figure}
  \begin{minipage}[h]{0.70\linewidth}
  \subfigure[]{
    \label{fig:cachesize-3d} 
    \includegraphics[width=0.50\textwidth]{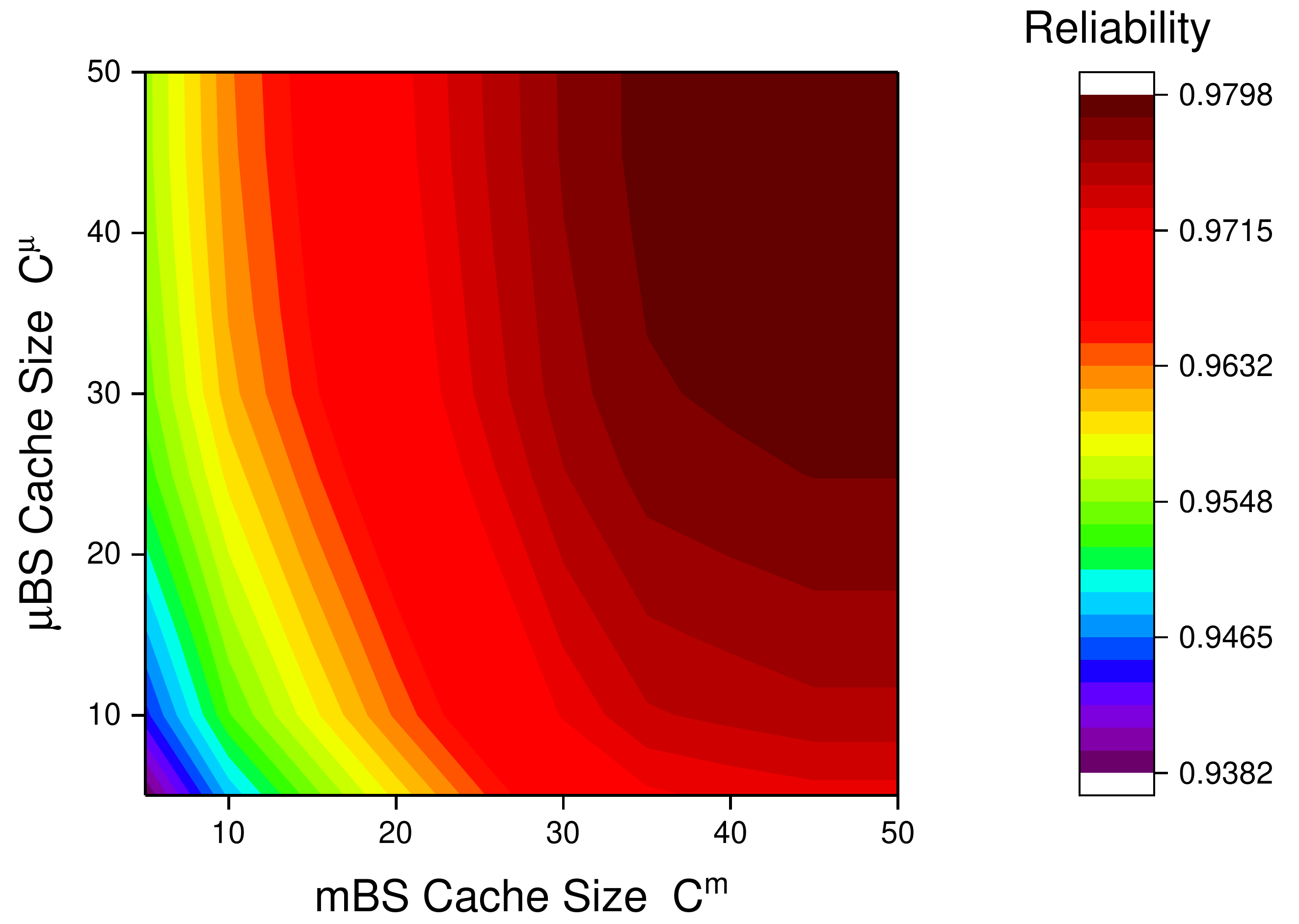}}
    \hspace{-3mm}
  \subfigure[]{
    \label{fig:fmufm-3d} 
    \includegraphics[width=0.49\textwidth]{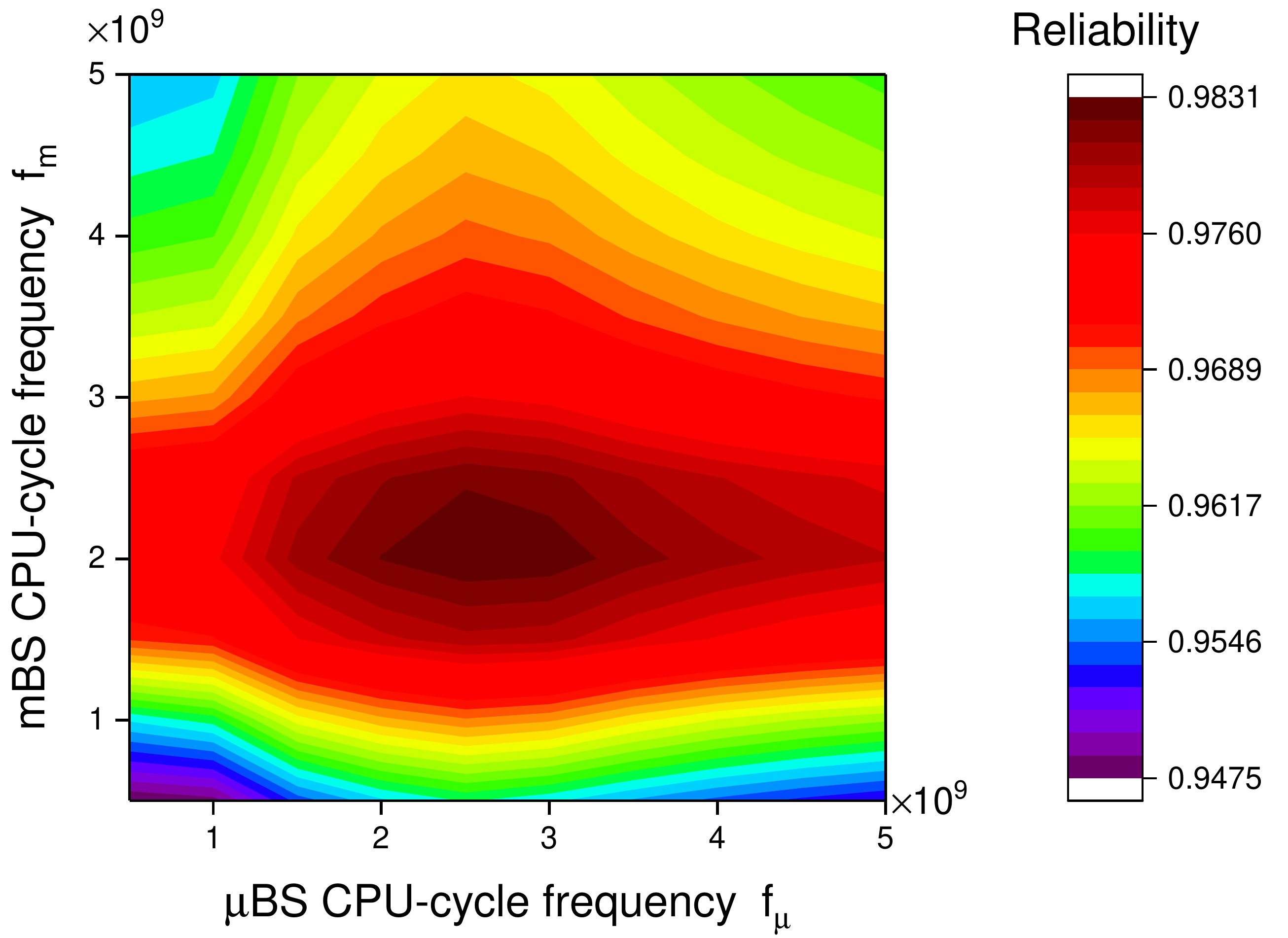}}
  \caption{(a) Reliability of VR delivery for various cache size of $\mu$BSs and mBSs, and (b) reliability of VR delivery for various CPU-cycle frequency of $\mu$BSs and mBSs. }
  \end{minipage}
  \hspace{0.04in}
  \begin{minipage}[h]{0.28\linewidth}\label{fig:4}
  \includegraphics[width=1.0\textwidth]{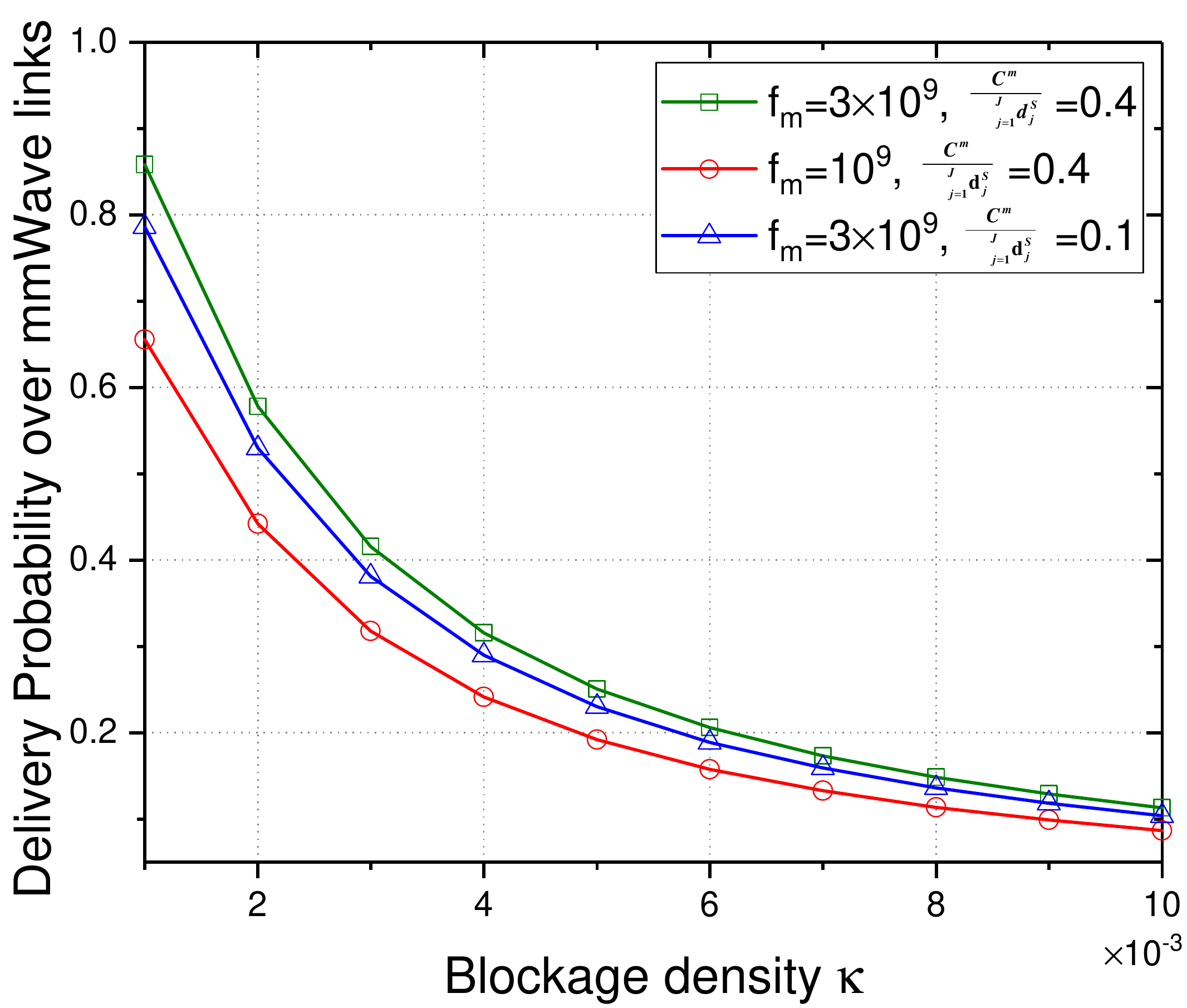}
  \caption{Delivery probability over mmWave links for various mmWave network parameters. }  %
  \end{minipage}
\end{figure}

The impact of caching capacities of $\mu$BSs and mBSs on the reliability of VR delivery is shown in Fig. \ref{fig:cachesize-3d}. In general, the reliability is higher with a larger caching capacity. It is observed that the caching capacity has a greater gain in the reliability of VR delivery over mmWave links. For example, the reliability is about 0.97 when $C^\mu=10, C^m=50$, while the reliability is about 0.95 when $C^\mu=50, C^m=10$. This is because the SVs cached at mBSs can be delivered over large-bandwidth mmWave links, so that the SVs  can be delivered to the HMD in time.

Fig. \ref{fig:fmufm-3d} shows the impact of CPU-cycle frequency reliability of $\mu$BSs and mBSs on the reliability of VR delivery. It is observed that, it is not that the higher the CPU-cycle frequency, the higher the reliability can be obtained. The reason is analyzed in conjunction with the energy consumption limitation  in the following subsection.  Note that the reliability of VR delivery is more sensitive to $f_m$ compared with $f_\mu$. In other words, the impact on the reliability is greater when $f_m$ is changed. This is because the mBSs needs to consume more computing resources to project and render the MVs into SVs, and deliver them over the  large-bandwidth mmWave links. In contrast, the number of MVs computed into SVs at $\mu$BSs is relatively small, which results in less impact on the reliability.

\subsection{Delivery Probability}

The delivery probability over mmWave links is validated with various network parameters as shown in Fig. 4. It can be observed  that the delivery probability over mBSs decreases with the increase of the blockage density due to the deteriorating channel conditions, which reflects the importance of \textcolor{blue}{sub-6 GHz} links to enhance the reliability of VR delivery despite the large mmWave bandwidth.
In addition, the delivery probability over mBSs increases  when the CPU-cycle frequency or the caching capacity of mBSs increases. This is because the increase in  caching  and computing  resources reduces the average end-to-end delay, thus more viewpoints are delivered over  mmWave links  according to the minimum delay criterion.

\begin{figure}
  \centering
  \subfigure[]{
    \label{fig:prop-caching-strategy} 
    \includegraphics[width=0.3\textwidth]{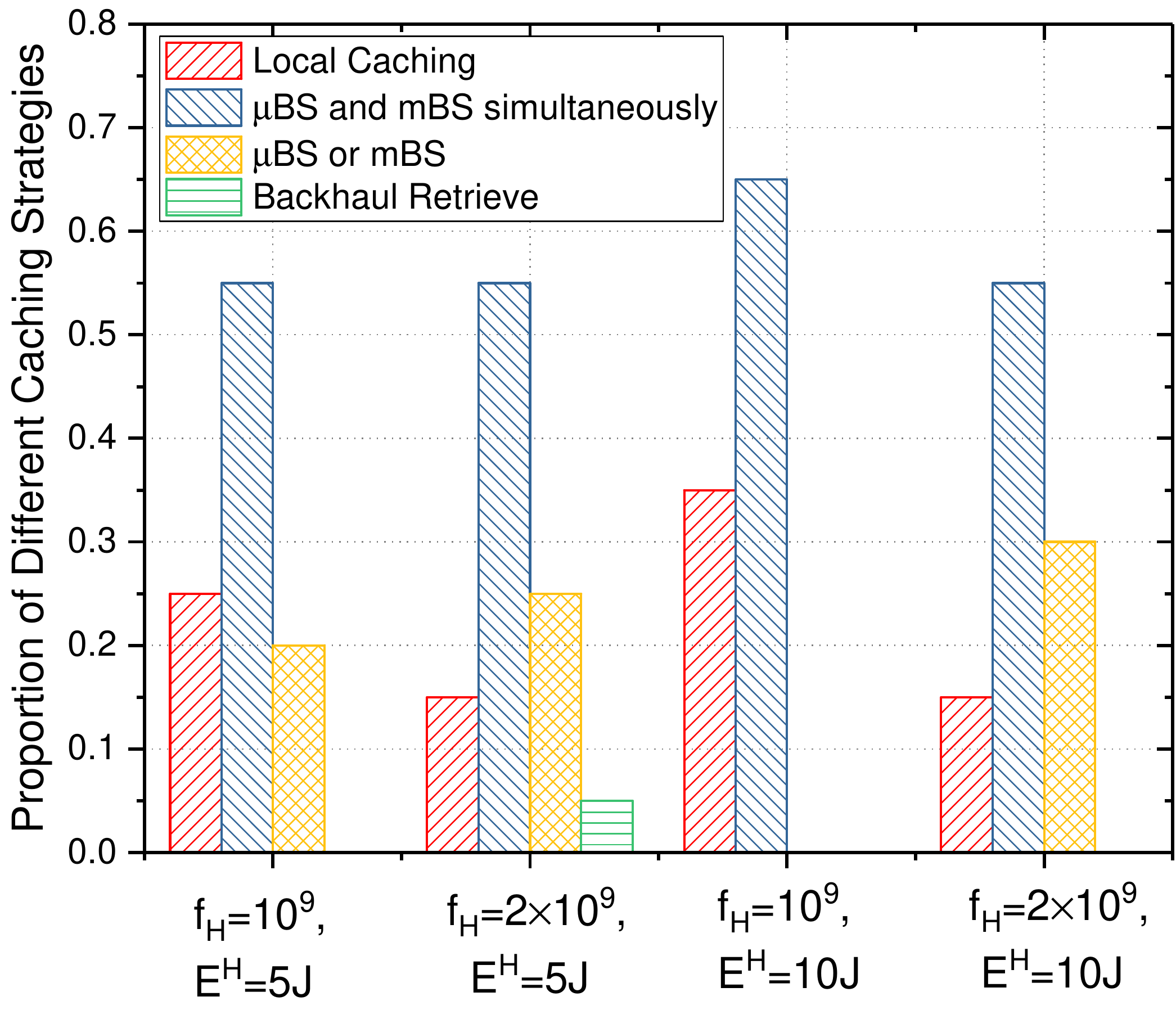}}
  \hspace{0.05in}
  \subfigure[]{
    \label{fig:prop-computing-strategy} 
    \includegraphics[width=0.3\textwidth]{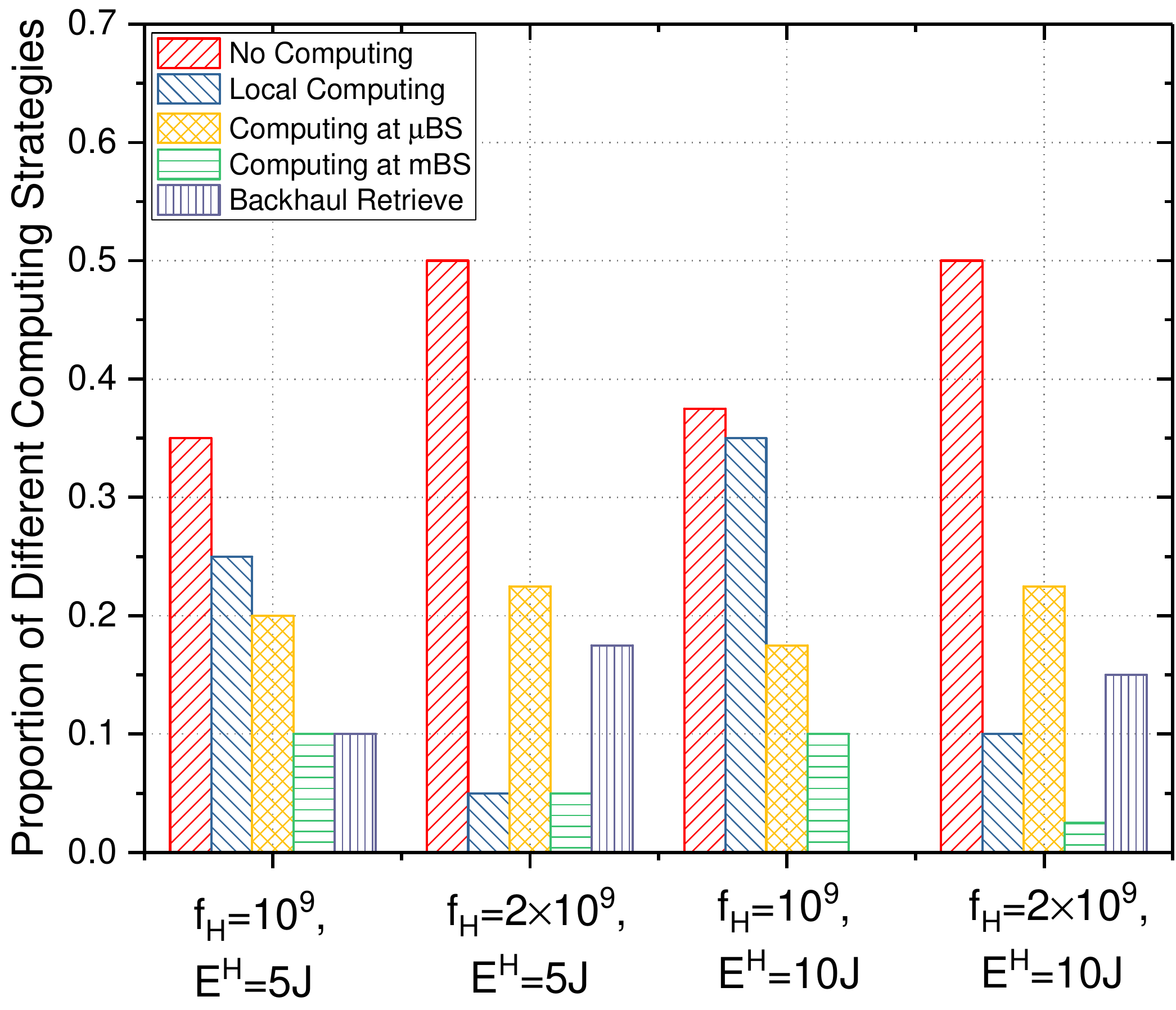}}
  \hspace{0.05in}
  \subfigure[]{
    \label{fig:prop-caching-2D3D} 
    \includegraphics[width=0.3\textwidth]{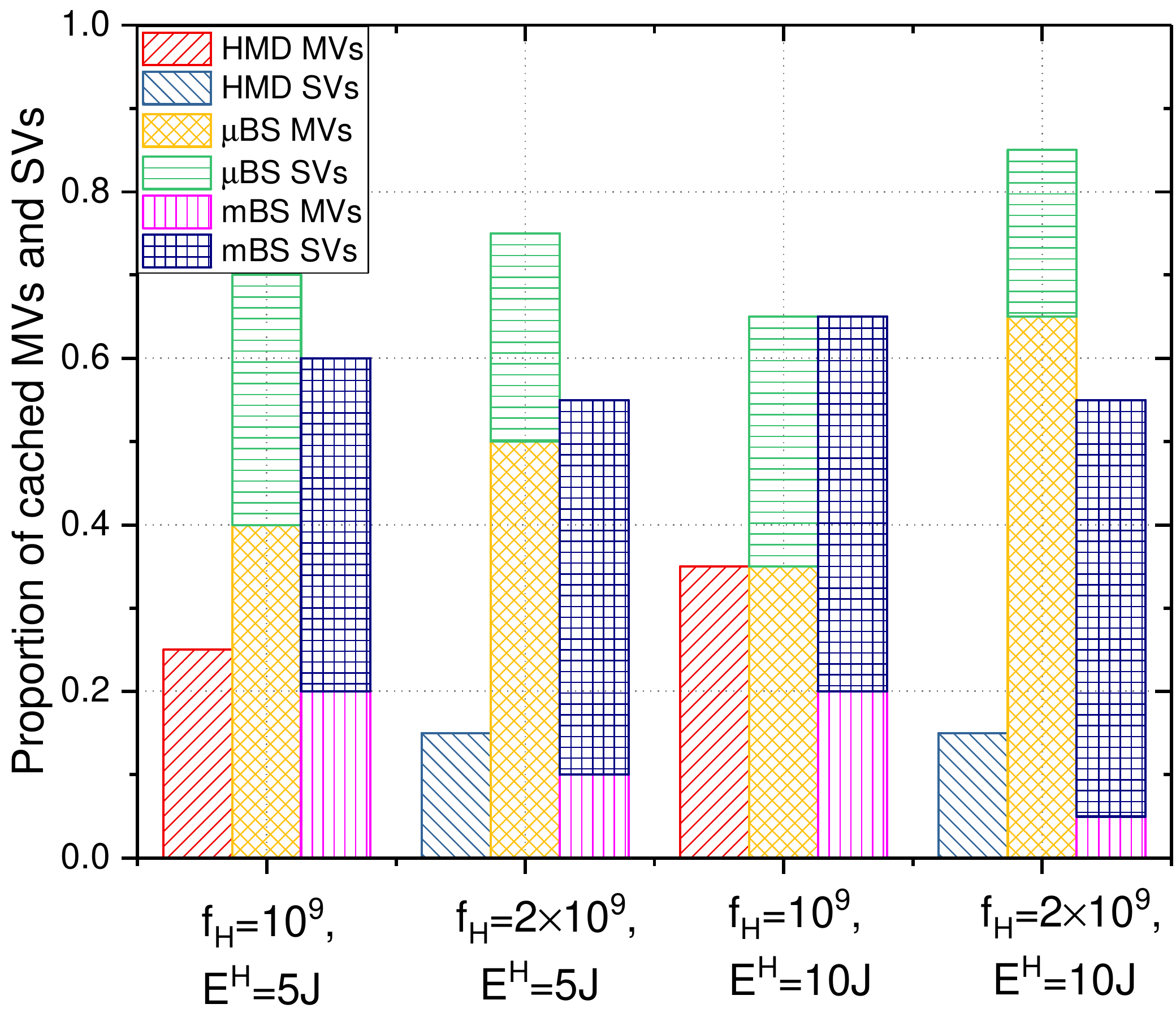}}
  \caption{(a) Proportion of different caching strategies under various $f_H$ and $E^H$, (b) proportion of different computing strategies under various $f_H$ and $E^H$, and (c) proportion of cached MVs and SVs under various $f_H$ and $E^H$. }
\end{figure}

\subsection{Caching and Computing Strategies}
We now evaluate the caching and computing strategies to obtain design insights.
Fig. \ref{fig:prop-caching-strategy} shows the proportion of different caching strategies under various $f_H$ and $E^H$. It can be seen that the locally cached viewpoints decrease when the local CPU-cycle frequency $f_H$ increases. This is due to the local energy consumption limitation. Higher computing frequency leads to higher energy consumption, as shown in Fig. \ref{fig:prop-computing-strategy}. Therefore, the proportion of local computing  should be reduced, otherwise it will cause HMD to heat up and reduce the user experience. Thus, more SVs are cached in the HMD. On the other hand, when $f_H$ is unchanged and $E^H$ is increased, the proportion of local caching or simultaneously cached in $\mu$BS and mBS increases. This is because the local energy consumption limitation  is loosen, and more MVs can be cached locally or at BSs. Then more MVs are delivered to the HMD over \textcolor{blue}{sub-6 GHz} links or mmWave links and computed into SVs locally, as shown in the proportion of local computing in Fig. \ref{fig:prop-computing-strategy}. It is worth noting that when $f_H=10^9$, the proportion of cached at $\mu$BS or mBS decreases with the increase of $E^H$. This is because the energy consumption limitation is strict at HMD  when $E^H$ is small, thus more SVs should be cached in the limited caching capacity at $\mu$BSs and mBSs to avoid the backhaul retrieve delay. When $E^H$ increases, more MVs can be cached, so more viewpoints can be cached simultaneously at $\mu$BSs and mBSs.

The proportion of cached MVs and SVs under various $f_H$ and $E^H$ is shown in Fig. \ref{fig:prop-caching-2D3D}. It can be seen that  SVs are mainly cached at mBSs, while MVs are  more cached at $\mu$BS. This indicates that it is beneficial to cache more SVs at mBSs and deliver them over mmWave links, while cache more MVs at $\mu$BSs  and deliver them  over \textcolor{blue}{sub-6 GHz} links. This is reasonable, because delivering more SVs over mmWave links will not increase the transmission delay much due to the large mmWave bandwidth, and the computing delay is saved. On the contrary, more MVs can be delivered over \textcolor{blue}{sub-6 GHz} links  due to the relatively low bandwidth.

\section{Conclusion}
In this paper, we propose a dual-connectivity \textcolor{blue}{sub-6 GHz} and mmWave HetNet architecture empowered by mobile edge capability, with the aim of improving the reliability of VR delivery.
Based on the differentiated characteristics of \textcolor{blue}{sub-6 GHz} links and mmWave links, we utilize their complementary advantages to conduct a collaborative design to improve the reliability of  VR delivery.
From the perspective of stochastic geometry,  we first derive  closed-form expressions for the reliability of VR delivery.
We propose a link selection strategy based on the minimum-delay delivery, and derive the VR delivery probability over  \textcolor{blue}{sub-6 GHz} links and mmWave links. We theoretically show the necessity of \textcolor{blue}{sub-6 GHz} links to improve the reliability despite the large mmWave bandwidth. We then formulate a joint caching and computing optimization problem to maximize the reliability of  VR delivery. By analyzing the coupling caching and computing strategies, we further transform the problem into a MMKP and propose a BFBB algorithm to obtain the optimal solution. To further reduce the complexity of the algorithm, we leverage DCP algorithm to obtain a sub-optimal solution. Numerical simulations demonstrate the performance improvement using the proposed algorithms, and shows great promise to improve the  VR delivery reliability in mobile-edge empowered DC \textcolor{blue}{sub-6 GHz} and mmWave HetNets.

In future work, a prospective direction is to improve the quality of experience (QoE) for VR users in DC \textcolor{blue}{sub-6 GHz} and mmWave HetNets. The QoE of VR video transmission is influenced by many factors such as the video rendering, quality of viewpoints, end-to-end delay, and delivery reliability, which is more challenging for modeling and optimization. A QoE-driven cross-layer design framework for mobile-edge empowered DC \textcolor{blue}{sub-6 GHz} and mmWave HetNets is anticipated, in which resource coordination that dynamically adapts to network conditions can be designed to achieve QoE enhancement.
%
\appendices
\section{Proof of Proposition 1} \label{AppendixA}
According to (\ref{SINR-mu}) and (\ref{rel-def}), the reliability of VR delivery over \textcolor{blue}{sub-6 GHz} links can be derived as
\begin{footnotesize}
\begin{align}\label{proof-muwave-reliability}
  \mathcal{R}_j^{\mu}(D_j^\mu, T_j^{\mu, t})  = \mathbb{P}\left[ \frac{P_{\mu} h_j^{\mu} r^{-\alpha_\mu} }{I_j^{\mu} + \sigma_{\mu}^2  } > \nu_j^{\mu} \right] & = \int_0^{\infty} \mathbb{P}\left[\frac{P_{\mu} h_j^{\mu} r^{-\alpha_\mu} }{I_j^{\mu} + \sigma_{\mu}^2  } > \nu_j^{\mu} \Big| r  \right] f_r(r) \mathrm{d}r, \nonumber  \\
  & \overset{(a)} = \int_0^{\infty} \mathrm{e}^{- \pi \lambda_\mu r^2} \mathrm{e}^{-\nu_j^{\mu} r^{\alpha_\mu} P_\mu^{-1} \sigma_{\mu}^2 } \mathcal{L}_{I_r}(\nu_j^{\mu} r^{\alpha_\mu} P_\mu^{-1}) \cdot 2 \pi \lambda_\mu r \mathrm{d} r,
\end{align}
\end{footnotesize}
$\!\!$where (a) follows by using the fact that $h_j^\mu $ obeys exponential distribution, and $\mathcal{L}_{I_j^\mu}(\nu_j^{\mu} r^{\alpha_\mu} P_\mu^{-1})$ is the Laplace transform of random variable $I_r$. Let $s = \nu_j^{\mu} r^\alpha_\mu P_\mu^{-1}$, we have
\begin{footnotesize}
\begin{align}
  \mathcal{L}_{I_j^\mu} (s)  = \mathbb{E}_{\Phi_\mu, h}[\mathrm{e}^{-s I_j^\mu}] = \mathbb{E}_{\Phi_\mu} \left[ \prod_{x \in \Phi_\mu \backslash b(o, r)} \frac{s \ell(x)}{1 + s \ell(x)} \right]  & \overset{(b)} = \exp \left( - 2 \pi \lambda_{\mu} \int_r^{\infty} \frac{sx}{s + x^{\alpha_{\mu}}} \mathrm{d} x \right), \nonumber \\
  & \overset{(c)} = \exp(\pi \lambda_\mu s^{\delta_\mu} \Gamma(1 + \delta_\mu) \Gamma(1 - \delta_\mu ) - \pi \lambda_\mu r^2 H_{\delta_\mu} (r^{\alpha_\mu}/s)), \nonumber
\end{align}
\end{footnotesize}

\vspace{-8mm}
$\!\!\!\!\!\!$where  $\ell(x) =  x^{-\alpha_\mu}$, $\delta_\mu = 2 / \alpha_\mu$. (b) follows from the probability generating functional (PGFL) of the PPP \cite{chiu2013stochastic}, and (c) follows from the use of gamma function or Gauss hypergeometric function for integration, where   $H_{\delta}(x) \triangleq {_2}F_{1}(1,\delta; 1+\delta; -x) $ is the Gauss hypergeometric function. By applying the Gauss-Laguerre Quadrature \cite{golub1969calculation}, the desired proof is obtained.

\vspace{-1mm}
\section{Proof of Proposition 2} \label{AppendixB}
According to (\ref{SINR-m}) and (\ref{rel-def}), the reliability of VR delivery over mmWave links is derived as
\vspace{-5mm}

\begin{small}
\begin{align}
  \mathcal{R}_j^{m}(D_j^m, T_j^{m, t})    = & \ \mathbb{P} \Bigg[\frac{P_{m} h_j^{m} G r^{- \alpha_m}}{ I_j^m + \sigma_{m}^2 } > \nu_j^{m} \Bigg], \nonumber \\
\overset{(d)} = & \ \int_{0}^{\infty} \sum_{i \in \{ \mathrm{L, N}\}} \rho_i(r) \left\{ 1- \mathbb{E}_{I_j^m} \left[ \left( 1-\mathrm{exp} \left( -\frac{\eta_i \nu_j^{m} r^{\alpha_i}  ( I_j^m + \sigma_{m}^2 ) }{P_{m} G}\right) \right)^{N_i} \right]  \right\} f_r(r)  \mathrm{d} r,  \nonumber \\
\overset{(e)} = & \ \int_{0}^{\infty} \sum_{i \in \{ \mathrm{L, N}\}} \rho_i(r) \left\{ \sum_{k=1}^{N_i} (-1)^{k+1} \binom{N_i}{k}  \mathrm{e}^{ -\frac{k \eta_i \nu_j^{m} r^{\alpha_i}  \sigma_{m}^2 }{P_{m} G} } \mathcal{L}_{I_j^m} \left( \frac{k \eta_i \nu_j^{m} r^{\alpha_i}  \sigma_{m}^2 }{P_{m} G} \right) \right\} f_r(r) \mathrm{d} r, \nonumber
\end{align}
\end{small}
\vspace{-5mm}

$\!\!\!\!\!\!$where (d) follows from the Alzer's approximation of a gamma random variable \cite{alzer1997some}. (e) follows by using Binomial theorem and the assumption that $N_i$ is an integer.
\vspace{-1mm}
\begin{small}
\begin{align}
   \mathcal{L}_{I_j^m} \left( s_i \right) \overset{(f)} = & \ \mathbb{E}_{I_j^m}  \left[ \prod_{{\ell \in \Phi_{m} \backslash b_m }} \mathbb{E}_h \left[ \mathrm{exp} \bigg( -s_{i}  h_\ell P_{m} G r^{-\alpha_{i}}\bigg)\right] \right], \nonumber  \\
   \overset{(g)} = & \ \mathrm{exp} \left[ -2 \pi  \lambda_{m} p_G \int_r^{\infty} \left( 1 - \mathbb{E}_{h} \left[ \mathrm{e}^{-s_{i} h_\ell P_{m} G t^{-\alpha_{i} }}  \right]  \right)  t \mathrm{d} t \right], \nonumber \\
    \overset{(h)} = & \ \mathrm{exp} \left[ -2 \pi  \lambda_{m} p_G \int_r^{\infty} \left( 1 - \frac{1}{\left(1 + s_{i} P_{m} G t^{-\alpha_i} / N_i\right)^{N_i}}\right) t \mathrm{d} t \right], \label{Laplace-mmWave}
\end{align}
\end{small}
$\!\!$where (f) follows from the i.i.d. distribution of $h$ and the independence of the PPP. (g) follows

$\!\!\!\!\!\!$by computing the PGFL of the PPP. (h) follows by computing the moment generating function of the gamma random variable $h$. Applying the integral formula of powers of $t$ and powers of binomials \cite{gradshteyn2014table}, (\ref{Laplace-mmWave}) can be written in the form of Gauss hypergeometric function, then the desired result is obtained.
\end{spacing}




\end{document}